\documentclass[journal]{IEEEtran}

\usepackage{mathptmx}       % selects Times Roman as basic font
\usepackage{helvet}         % selects Helvetica as sans-serif font
\usepackage{courier}        % selects Courier as typewriter font
\usepackage{type1cm}        % activate if the above 3 fonts are
\usepackage{makeidx}         % allows index generation
\usepackage{graphicx}        % standard LaTeX graphics tool

\usepackage{multicol}        % used for the two-column index
\usepackage[bottom]{footmisc}% places footnotes at page bottom
\usepackage{ifthen}
\usepackage{subfigure}
\usepackage[usenames,dvipsnames]{color}

\makeindex             % used for the subject index
\usepackage{graphics} % for pdf, bitmapped graphics files
\usepackage{graphicx} % for pdf, bitmapped graphics files
\usepackage{epsfig} % for postscript graphics files
\usepackage{epstopdf}
\usepackage{mathptmx} % assumes new font selection scheme installed
\usepackage{times} % assumes new font selection scheme installed
\usepackage[cmex10]{amsmath}

\usepackage{mathtools}  % loads amsmath
\usepackage{amssymb}  % assumes amsmath package installed
\usepackage{amsthm}

\usepackage{amsfonts}

\usepackage{subfig}
\usepackage{verbatim}
\usepackage{comment}
\usepackage{algorithm}
\usepackage{algorithmic}
\usepackage{multirow}
\usepackage{cite}

\DeclareGraphicsExtensions{.pdf,.png,.jpg,.eps}

\def\ba{\begin{array}}
\def\ea{\end{array}}
\def\baa{\begin{align}}
\def\eaa{\end{align}}
\newcommand{\beq}{\begin{equation}}
\newcommand{\eeq}{\end{equation}}
\newcommand{\bq}{\begin{eqnarray}}
\newcommand{\eq}{\end{eqnarray}}
\newcommand{\bqn}{\begin{eqnarray*}}
\newcommand{\eqn}{\end{eqnarray*}}
\newcommand{\bee}{\begin{enumerate}}
\newcommand{\eee}{\end{enumerate}}
\newcommand{\bi}{\begin{itemize}}
\newcommand{\ei}{\end{itemize}}

\newcommand{\diag}{\mathrm{diag}}

\newboolean{showcomments}
\setboolean{showcomments}{true}
\newcommand{\wang}[1]{\ifthenelse{\boolean{showcomments}}
{ \textcolor[rgb]{1,0,1}{(ZW:  #1)}}{}}
\newcommand{\fliu}[1]{\ifthenelse{\boolean{showcomments}}
{ \textcolor{red}{(FL:  #1)}}{}}
\newcommand{\zhao}[1]{\ifthenelse{\boolean{showcomments}}
{ \textcolor{green}{(CZ:  #1)}}{}}
\newcommand{\slow}[1]{\ifthenelse{\boolean{showcomments}}
{ \textcolor{blue}{(SL:  #1)}}{}}

\theoremstyle{definition}
\newtheorem{theorem}{Theorem}
\newtheorem{lemma}[theorem]{Lemma}

\theoremstyle{definition}
\newtheorem{definition}{Definition}
\newtheorem{remark}{Remark}

\ifCLASSOPTIONcaptionsoff
\usepackage[nomarkers]{endfloat}
\let\MYoriglatexcaption\caption
\renewcommand{\caption}[2][\relax]{\MYoriglatexcaption[#2]{#2}}
\fi

\title{Distributed Frequency Control with Operational Constraints, Part II: Network Power Balance}

        \author{Zhaojian~Wang, Feng~Liu, Steven~H.~Low,~\IEEEmembership{Fellow,~IEEE},
                Changhong~Zhao, and~Shengwei~Mei~\IEEEmembership{Fellow,~IEEE}
                \thanks{This work was supported  by the National Natural Science Foundation
                        of China ( No. 51377092, No. 51677100, No. 51621065), Foundation of Chinese Scholarship Council (CSC No. 201506215034), 
the US National Science Foundation through awards EPCN 1619352, CCF 1637598, 
CNS 1545096, ARPA-E award DE-AR0000699, and Skoltech through Collaboration
                        Agreement 1075-MRA.     }       % <-this % stops a space
                \thanks{Z. Wang, F. Liu and S. Mei are with the Department
                        of Electrical Engineering, Tsinghua University, Beijing,
                        China, 100084 e-mail: (lfeng@tsinghua.edu.cn).}% <-this % stops a space
                \thanks{S. H. Low and C. Zhao  are with the Department
                        of Electrical Engineering, California Institute of Technology, Pasadena, CA, USA, 91105 e-mail:(slow@caltech.edu)}
        }

\begin{document}

         \maketitle

\begin{abstract}                                                
    In Part I of this paper we propose a decentralized optimal frequency control of multi-area power system with operational constraints, where the tie-line powers  remain unchanged in the steady state and the power mismatch is balanced within individual control areas. In Part II of the paper, we propose a distributed controller for optimal frequency control in the network power balance case, where the power mismatch is balanced over the whole system. With the proposed controller,
	the tie-line powers remain within the acceptable range at equilibrium, while the regulation capacity constraints are satisfied both at equilibrium and during transient. It is revealed that the closed-loop system with the proposed controller carries out primal-dual updates with saturation for solving an associated optimization problem. To cope with discontinuous dynamics of the closed-loop system, we deploy the invariance principle for nonpathological Lyapunov function to prove its asymptotic stability. Simulation results are provided to show the effectiveness of our controller.
\end{abstract}
                

\IEEEpeerreviewmaketitle

\section{Introduction}
In Part I of the paper we have investigated the optimal frequency control of multi-area power system with operational constraints \cite{Wang:DistributedFrequency}. In that case, the tie-line powers are required to be unchanged in the steady state after load disturbances, which implies the power mismatch in each area has to be balanced individually. It is referred to as the per-node power balance case. In Part II of the paper, we consider the transmission congestion in the distributed optimal frequency design. 

The per-node balance case in Part I mainly considers the situation where the power delivered from one area to another is fixed, e.g. contract power, which should not be violated in normal operation. However, in some circumstances, control areas may cooperate for better frequency recovery or regulation cost reduction. In this case, power mismatch may be  balanced by all generations and controllable loads among all control areas in cooperation. Similar situations also appear in one control area with multiple generators and controllable loads that cooperate to eliminate the power mismatch in the area. It is referred to as the network power balance case. Compared with the per-node balance case, the most challenging problem in this case is that the tie-line powers may change and congestions may occur. In addition, local information is not sufficient and neighboring information turns to be helpful. As for the constraints, the tie-line power constraints are not hard limits, which only need to be satisfied at equilibrium. The capacity limits on the generations and  controllable loads also required to be satisfied both in  steady state and during transient.

In the recent  literature of frequency control, tie-line power constraints are considered in \cite{Mallada-2017-OLC-TAC, Stegink:aunifying, Stegink:aport,Yi:Distributed, zhao:aunified, zhang:real}. In \cite{Mallada-2017-OLC-TAC},  tie-line power constraints are included in the load-side secondary frequency control. A virtual variable is used to estimate the tie-line power, whose value is identical to the tie-line power at equilibrium. In \cite{Stegink:aunifying, Stegink:aport}, an optimal economic dispatch problem including  tie-line power constraints is formulated, then the solution dynamics derived from a primal-dual algorithm is shaped as a port-Hamiltonian form. The power system dynamics also have a port-Hamiltonian form, which are interconnected with the solution dynamics to constitute a closed-loop Hamiltonian system. Then, the optimality and stability are proved. In \cite{zhao:aunified}, a unified method is proposed for  primary and secondary frequency control, where the congestion management is implemented in the secondary control. In \cite{zhang:real}, a real-time control framework is proposed for tree power networks, where transmission capacities are considered. 

% Moreover the closed-loop system remains asymptotically stable and converges to the unique optimal point.

% \fliu{The current presentation looks quite clear for me. Should we mention the \emph{approximate} primal-dual in Introduction?}
% 
% \slow{Comparison of per-node and network case is in Remark 3.  Refer to \cite{LiZhaoChen2016}.}
 
Similar to the per-node balance case, hard limits, such as capacity constraints of power injections on buses, are enforced only in the steady-state in the literature, which may fail if such constraints are violated in transient. Here we construct a fully distributed control to recover nominal frequency while eliminating congestion. Differing from the literature, it enforces regulation capacity constraints not only at equilibrium, but also during transient. We show that the controllers together with the physical dynamics serve as  primal-dual updates with saturation for solving the optimization problem. 
The optimal solution of the optimization problem and the equilibrium point of closed-loop system are identical. 

The  enforcement of capacity constraints during transient and tie-line power limits simultaneously makes the stability proof difficult. Specifically, the Lyapunov function  is not continuous anymore, as in the per-node case in Part I of the paper. In this situation, the conventional LaSalle's invariance principle does not apply. To overcome the difficulty, we construct a nonpathological Lyapunov function to mitigate the impacts of nonsmooth dynamics. The salient features of the controller are: 
\bee
\item \emph{Control goals:} the controller restores the nominal frequency and balance the power mismatch in the whole system after unknown load disturbance while minimizing the regulation costs; 

\item \emph{Constraints:} the regulation capacity constraints are always enforced even during transient and the congestions can be eliminated automatically; 

\item \emph{Communication:} only neighborhood communication is needed in the  network balance case;

% \slow{Comment that per-node balance case means that
%no communication is needed between control areas.  It does not imply that the generators within
%the same area will not need communication to restore the nominal frequency.
%Check and cite MIT algorithm that requires no communication to restore nominal frequency.}
%
%\fliu{You're right. Within a control area, there should be a requirement of communication. But I am not sure which paper you mean...maybe Zhaojian know it? }
%
%\wang{I find two papers may be related to the problem,"Distributed Frequency Control in Power Grids under limited communication" and "Modeling the Impact of Communication Loss on the power grid under emergency control". Both of them deal with the communication failure problem. Are they right? I think it is not necessary to mention communication within one area as we use the aggregate model. How about if the reviewer wants to know, we cite them then? }

\item \emph{Measurement:} the controller is adaptive to unknown load disturbances automatically with no need of load measurement. 
\eee

The rest of this paper is organized as follows. In Section II, we describe our model. Section III formulates the optimal frequency control problem in the network balance case, presents the distributed  frequency controller and proves  the optimality, uniqueness and stability of the closed-loop equilibrium. Simulation results are given in Section IV. Section V concludes the paper.

\section{Network model}

%\subsection{Network model}

We summarize the notation used in Part I \cite{Wang:DistributedFrequency}. The power network is model by a directed graph ${G}:=(N, E)$ 
where  $N=\{0,1,2,...n\}$ is the set of nodes (control areas) and
$E\subseteq N\times N$ is the set of edges (tie lines).  If a pair of
nodes $i$ and $j$ are connected by a tie line
directly, we denote the tie line by $(i,j)\in E$.
Let $m:= |E|$ denote the number of tie lines.
Use $(i,j)\in E$
or $i\rightarrow j$ interchangeably to denote a directed edge from $i$ to $j$.
Assume the graph is connected
and node $0$ is a reference node.
For each node $j\in N$, $\theta_j(t)$ denotes the rotor angle at node $j$ at time
$t$ and $\omega_j(t)$ is the  frequency.
$P_j^g(t)$ denotes the (aggregate) generation at node
$j$ at time $t$ and $u^g_j(t)$ is its generation control command.
$P^l_j(t)$ denotes the (aggregate) controllable load and $u^l_j(t)$ is its
load control command.  $p_j$ is the (aggregate) uncontrollable load. 

%We adopt a second-order linearized model to describe the frequency dynamics
%of each node, and two first-order inertia equations to describe the dynamics of
%power generation regulation and load regulation at each node.  We assume
%the tie lines are lossless and adopt the DC power flow model.
The power system dynamics for each node $j\in N$ is
\begin{subequations}
	\begin{align}
	\dot \theta_j & =  \omega_j(t)
	\label{eq:model.1a}
	\\
	M_j \dot \omega_j & =   P^g_j(t) - P^{l}_j(t) - p_j -D_j \omega_j(t)
	\nonumber
	\\
	&  + \sum_{i: i\rightarrow j} \! B_{ij} (\theta_i(t) - \theta_j(t))
	-  \sum_{k: j\rightarrow k} \! B_{jk}(\theta_j(t) - \theta_k(t))
	\label{eq:model.1b}
	\\
	T^g_j \dot P^g_j & =  - P^g_j(t) + u^g_j(t) - {\omega_j(t)}/{R_j}
	\label{eq:model.1c}
	\\
	T^l_j \dot P^{l}_j & =  - P^{l}_j(t) + u^l_j(t)
	\label{eq:model.1d}
	\end{align}
	where $D_j>0$ are damping constants, $R_j>0$ are droop parameters,
	and $B_{jk}>0$ are line parameters that depend on the reactance of the line $(j,k)$.
	Let  $x := (\theta, \omega, P^g, P^l)$ denote the state of the network
	and $u := (u^g, u^l)$ denote the control.
%	\footnote{Given
%		a collection of $x_i$ for $i$ in a certain set $A$, $x$ denotes the column vector
%		$x := (x_i, i\in A)$ of a proper dimension with $x_i$ as its components.}
	\label{eq:model.1}
\end{subequations}

%Our goal is to design feedback control laws for the generation command
%$u^g(x(t))$ and load control $u^l(x(t))$.
The capacity constraints are:
\begin{subequations}
        \bq
                \underline{P}^g_j & \leq \ P^g_j(t) \ \leq \overline{P}^g_j, \quad j\in N
        \label{eq:OpConstraints.1a}\\
                \underline{P}^l_j & \leq \ P^l_j(t) \ \leq \overline{P}^l_j, \quad j\in N
        \label{eq:OpConstraints.1b}
        \eq        
        Here \eqref{eq:OpConstraints.1a}  and \eqref{eq:OpConstraints.1b} are  \emph{hard limits} on the regulation capacities of generation and controllable load at each node, which should not be violated at any time even during transient.
     \label{eq:OpConstraints.1}
\end{subequations}
%Hence these constraints are satisfied not only at equilibrium, but also during transient.

The system operates in a steady state initially, i.e., the generation
and the load are balanced and the frequency is at its nominal value. All variables
represent {deviations} from their nominal or scheduled values so that,
e.g., $\omega_j(t)=0$ means the frequency is at its nominal value. 

In this paper, all nodes cooperate to rebalance power over the entire network after a disturbance.  The power flows $P_{ij}$ on the tie lines may deviate from their scheduled values and we require that they satisfy line limits,
i.e., 
\bqn
\underline{P}_{ij} \ \le  \ P_{ij} \ \le \ \overline{P}_{ij} \qquad\quad \forall (i,j)\in E
\eqn
for some upper and lower bounds $\underline{P}_{ij}, \overline{P}_{ij}$. 

In DC approximation the power flow on line $(i,j)$ is given by 
$P_{ij} =B_{ij}(\theta_i - \theta_j)$.  Hence line flow constraints in the per-node
balance case is $\theta_i = \theta_j$ for all $(i,j)\in E$ and in the network balance
case is: 
\bq
\label{eq:lineConstraint}
\underline{\theta}_{ij} \ \le \ \theta_{i} - \theta_j \ \le \ \overline{\theta}_{ij} \qquad\quad \forall (i,j)\in E
\eq
where $\underline{\theta}_{ij}=\underline{P}_{ij}/B_{ij}$, $\overline{\theta}_{ij}=\overline{P}_{ij}/B_{ij}$ 
are lower and upper bounds on angle differences. 

As the generation $P^g_j$ and load $P^l_j$ in each area can increase or decrease, 
and a line flow $P_{ij}$ can in either direction, we make the following assumption.
\bi
\item[\textbf{A1:}] 
\bee
\item $\underline{P}^g_j < 0 < \overline{P}^g_j$ and $\underline{P}^l_j <0 < \overline{P}^l_j$
for $\forall j\in N$.
\item $\underline{\theta}_{ij} \leq 0 \leq \overline{\theta}_{ij}$ for $(i,j)\in E$.
\item $\theta_0(t):=0$ and $\phi_0(t):=0$ for all $t\geq 0$.
\eee
\ei  
Here $\phi$ is a vector variable that represents virtual phase angles in the
network balance case in Section \ref{sec:npb}.
The assumption $\theta_0 \equiv \phi_0 \equiv 0$ amounts to using $(\theta_0(t), \phi_0(t))$ as reference
angles.  It is made merely for notational convenience: as we will see, the equilibrium point will be
unique with this assumption (or unique up to reference angles without this assumption).

\section {Controller Design for The Network Power Balance Case}
\label{sec:npb}

In the per-node balance case, individual control areas rebalance power within their own areas
after disturbances. However, in many circumstances, it may be more efficient for all control 
areas to eliminate power imbalance of the overall system in a coordinated manner. This can 
be modeled as the condition:
\bq
\label{eq:balance.net}
\sum\nolimits_j P^g_j & = &  \sum\nolimits_j \left( P^l_j + p_j \right)
\eq
In this case the tie-line flows may not be restored to their pre-disturbance values.
To ensure that they are within operational limits, the constraints \eqref{eq:lineConstraint}
are imposed.

Even though the philosophy of the controller design as well as the proofs are similar to the per-node case, 
the details are much more complicated.   
Our presentation will however be brief where there is no confusion.

%In this situation, we have the following assumption
%\bi
%\item[A4:]  For $\forall j\in N$, $\sum\nolimits_j(\underline{P}^g_j - \overline{P}^l_j )<\sum\nolimits_j p_j$, $\sum\nolimits_j(\overline{P}^g_j - \underline{P}^l_j ) > \sum\nolimits_j p_j $.
%In addition, in the steady state, for $\forall j\in N$, the generation ($\underline P_j^g, \overline P_j^g$) and controllable load constraints ($\underline P_j^l, \overline P_j^l$)  cannot be reached  with the powers of tie lines connected to area $j$ ($\underline P_{ij}, \overline P_{ij}$, $i\in N_j$)  reaching their limits simultaneously.
%\ei 

%Similar to A1, assumption A4 is to guarantee that: 1) the network should has power margin;2) each area should has power margin when balances the load demands. Otherwise, the network or some area is small signal unstable.

\subsection{Control goals}
In the network power balance case, the control goals are formalized as the following 
optimization problem.
\begin{subequations}
        \bq
        \text{NBO:}\min &\!\!\!\!\!\!& \sum \limits_j \frac{\alpha_j}{2} \left(P^g_j\right)^2 
        +  \sum \limits_j \frac{\beta_j}{2} \left(P^l_j\right)^2
% \nonumber\\     
%         & \!\!\!\!\!\! &
         +  \sum \limits_j \frac{D_j}{2} \omega_j^2 +  \sum \limits_j \frac{z_j^2}{2}
          \nonumber\\     
        \label{eq:opt.2o}
        \\
        \text{over} & \!\!\!\!\!\! & x := (\theta, \phi, \omega, P^g, P^l) \text{ and }
        u := (u^g, u^l)
        \nonumber
        \\ 
        \text{s. t.}  
        & \!\!\!\!\!\! & \eqref{eq:OpConstraints.1},   
        \nonumber
        \\
        & \!\!\!\!\!\! & 
        P^g_j = P^l_j  + p_j +U_j(\theta, \omega) \quad j\in N
        \label{eq:opt.2a}
        \\
        & \!\!\!\!\!\! & 
        P^g_j = P^l_j  + p_j +\hat{U}_j(\phi) \quad j\in N
        \label{eq:opt.2b}
        \\
        & \!\!\!\!\!\! & \underline{\theta}_{ij} \le  \phi_i-\phi_j \le \overline{\theta}_{ij}
        \label{eq:opt.2c}, \quad (i,j) \in E 
        \\
        & \!\!\!\!\!\! & P^g_j \ = \ u^g_j, \quad j\in N
        \label{eq:opt.2d}
        \\
        & \!\!\!\!\!\! & P^l_j \ = \ u^l_j, \quad j\in N
        \label{eq:opt.2e}
        \eq
where $\alpha_j>0$, $\beta_j>0$ are constant weights; $z_j$ is a shorthand defined for convenience as        
\bqn
        z_j & := & P^g_j-P^l_j-p_j-\hat{U}_j(\phi)
\eqn
$U(\theta, \omega) := D\omega + CBC^T\theta$, and $\hat{U}(\phi)  :=  CBC^T \phi$.

While $\theta$ represents the phase angles in the physical power network,
$\phi$ is a cyber quantity that can be interpreted as virtual phase angles
(see remarks below).
The matrices $D$, $C$ and $B$ are defined as in the previous section.
\label{eq:opt.2}
\end{subequations}

As in the per-node  case, we define the variables $\tilde\theta_{ij}:=\theta_i - \theta_j$,
or in vector form, $\tilde\theta := C^T\theta$.
As we fix $\theta_0 := 0$ to be a reference angle under assumption A1, 
$\tilde\theta = C^T\theta$ defines a bijection between $\theta$ and $\tilde\theta$. 
Similarly we define 
$\tilde\phi_{ij}:=\phi_i - \phi_j$ or $\tilde{\phi}:=C^T\phi$, and $\phi_0 := 0$ so 
there is a bijection between $\phi$ and $\tilde\phi$.  
Note that both $\tilde{\theta}$ and $\tilde \phi$ are restricted to the column space of $C^T$. 
We will use $(\theta, \phi)$  and $(\tilde{\theta}, \tilde{\phi})$ interchangeably.
For instance we will abuse notation and write $\hat U(\phi):=CBC^T \phi$ or 
$\hat U(\tilde\phi):= CB\tilde \phi$.

% Note that we have abused and use the notation $\phi_{ij}=\phi_i-\phi_j$. 

We now summarize some of the interesting properties of NBO \eqref{eq:opt.2}
that will be proved formally in the next two subsections.  
We first compare NBO \eqref{eq:opt.2} with PBO in \cite{Wang:DistributedFrequency} for
the per-node balance case.
\begin{remark}[Comparison of NBO and PBO]
        \bee
        \item  Intuitively the network balance condition \eqref{eq:balance.net} is a relaxation 
        of the per-node balance condition (3) in \cite{Wang:DistributedFrequency}, and hence we
        expect that the optimal cost of NBO lower bounds that of PBO.   This is indeed the case, as we now argue.
		Constraint \eqref{eq:opt.2b} implies that any
        feasible point of \eqref{eq:opt.2} has $z_j=0$ and hence these two optimization
        problems have the same objective function.  
        % Both aims to recover the nominal frequency while minimizing the regulation costs.  
        Their variables and constraints 
        are different in that  PBO directly enforces the per-node balance  
        condition while NBO \eqref{eq:opt.2} has the additional 
        variable $\phi$ and constraints \eqref{eq:opt.2b}\eqref{eq:opt.2c}.
        Any optimal point $(\theta^*, \omega^*, P^{g*}, P^{l^*})$ for PBO
        however defines a feasible point $(\theta^*, \phi, \omega^*, P^{g*}, P^{l*})$ 
        for NBO \eqref{eq:opt.2} with the same cost where $\phi=\theta^*$.
        The point $(\theta^*, \phi, \omega^*, P^{g*}, P^{l*})$ satisfies 
        \eqref{eq:opt.2b}\eqref{eq:opt.2c} 
        because $(\theta^*, \omega^*, P^{g*}, P^{l^*})$ satisfies \eqref{eq:opt.2a},
        $\omega^*=0$ and $\theta^*_i = \theta^*_j$ by Theorem 2 in \cite{Wang:DistributedFrequency}, and
        $\underline{\theta}_{ij} \leq 0 \leq \overline{\theta}_{ij}$ by assumption A1.
     
        \item
        Even though any feasible point of \eqref{eq:opt.2} has $z_j=0$, the 
        objective function is augmented with $z_j^2$ to improve convergence 
        (see \cite{Feijer:Stability}).         
        \eee
\end{remark}

Even though neither the network balance condition \eqref{eq:balance.net} nor the
line limits \eqref{eq:lineConstraint} are explicitly enforced in \eqref{eq:opt.2}, 
they are satisfied at optimality (Theorem \ref{thm:6} below).
Indeed, the virtual phase angles $\phi$  
and the conditions \eqref{eq:opt.2a}--\eqref{eq:opt.2c} are carefully designed to
enforce these conditions as well as to restore the nominal frequency $\omega^*=0$
at optimality.  This technique is previously used in \cite{Mallada-2017-OLC-TAC}.
\begin{remark}[Virtual phase angles $\phi$]
\bee
        \item
        Under mild conditions, $\omega^*=0$ at optimality for both PBO and NBO.  
%        For PBO, this is a consequence of the per-node power balance requirement
%        \eqref{eq:balance.node} and the constraint \eqref{eq:opt.1a}; see Lemma 
%        \ref{lemma:1}.
        For NBO, this is a consequence of the constraint \eqref{eq:opt.2b} on $\phi$; 
        see Lemma \ref{lemma:6}.
        Summing \eqref{eq:opt.2b} over 
        all $j\in N$ also implies the network  balance condition \eqref{eq:balance.net}
        since $\textbf{1}^T \hat U(\phi) = \textbf{1}^T CBC^T\phi = 0$.

        \item
        In PBO, $\theta^*_i = \theta^*_j$ at optimality (i.e., tie-line flows are restored
        $P_{ij}^*=0$) and $U(\theta^*, \omega^*)=0$.  This does not necessarily hold in NBO.  
        However $\phi$ is regarded as virtual phase angles because, at optimality, 
        $\phi^*$ differs from the real phase angles $\theta^*$ only by a constant, 
        $\phi^*-\theta^*=\textbf{1}(\phi_0-\theta_0)$ (Lemma \ref{lemma:6} in
        the appendix).
        Hence $C^T\phi^*-C^T\theta^*=C^T\cdot\textbf{1}(\phi_0-\theta_0)=0$,  implying 
        $\tilde\phi_{ij}^*=\tilde\theta_{ij}^*$. 
        Then the constraints \eqref{eq:opt.2c} are exactly the flow constraints \eqref{eq:lineConstraint}.
        In other words, we impose the flow constraints on $\tilde{\theta}$ indirectly by enforcing such
        constraints on the virtual angle  $\tilde{\phi}$.  .
\eee
\end{remark}

\subsection{Distributed controller}

Our  control laws are: 
\begin{subequations}
        \bq
        \dot \lambda_j & = &  \gamma^{\lambda}_j \left( P^g_j(t) -  P^l_j(t) - p_j-\hat{U}_j(\tilde\phi(t))\right),
        \  j\in N
\label{eq:control.2a}\\
    \dot {\eta}^+_{ij}&=&\gamma^{\eta}_{ij}[\tilde\phi_{ij}(t)-\overline{\theta}_{ij}]^+_{{\eta}^+_{ij}}\qquad\quad\ \qquad \forall (i,j)\in E
\label{eq:control.2b}\\
    \dot {\eta}^-_{ij}&=&\gamma^{\eta}_{ij}[\underline{\theta}_{ij}-\tilde\phi_{ij}(t)]^+_{{\eta}^-_{ij}}\qquad\quad\ \qquad \forall (i,j)\in E
\label{eq:control.2c}\\
        \dot{\tilde\phi}_{ij} &= & \gamma^{\tilde\phi}_{ij}\left( B_{ij}[\lambda_i(t)-\lambda_j(t)+z_i(t)-z_j(t)]\right.
        \nonumber\\
        &&\left.  \ + \ \eta^-_{ij}(t)-\eta^+_{ij}(t)\right)\qquad\qquad\ \forall (i,j)\in E
\label{eq:control.2d}\\
        u^g_j(t) & = & \left[ 
        P^g_j(t) - \gamma^g_j \left( \alpha_j P^g_j(t) + \omega_j(t)+z_j(t) + \lambda_j(t) \right) \right]_{\underline P^g_j}^{\overline P^g_j}
        \nonumber\\
        &&  \ \, + \ {\omega_j(t)}/{R_j}, \qquad\qquad\qquad\ \ \, j\in N
\label{eq:control.2e}\\
        u^l_j(t) & = & \left[ 
        P^l_j(t) - \gamma^l_j \left( \beta_j P^l_j(t) - \omega_j(t)-z_j(t) - \lambda_j(t)\right)
        \right]_{\underline P^l_j}^{\overline P^l_j}
         \nonumber\\
& & \qquad\qquad\qquad\qquad\qquad\qquad\  j\in N		\label{eq:control.2f}
        \eq
where $\gamma^{\lambda}_j, \gamma^{\eta}_{ij}, \gamma^{\tilde\phi}_{ij}, \gamma^g_j, \gamma^l_j$ are positive constants. For any $x_i, a_i \in \mathbb R$, the operator $[x_i]^+_{a_i}$  is defined by 
        \label{eq:control.2}
\end{subequations}
\bqn
[x_i]^+_{a_i}&:=&\left\{
\begin{array}{ll}
x_i& \text{if} \ a_i>0 \ \text{or} \  x_i>0;\\
0,& \text{otherwise}.
\end{array}
 \right.
\eqn

For a vector case, $[x]^+_a$ is defined accordingly componentwise \cite{cherukuri:asymptotic}.

Here we assume that each node $i$ updates a set of internal states
$(\lambda_i(t), \eta^+_{ij}(t), \eta^-_{ji}(t), \tilde\phi_{ij}(t) )$ according
to \eqref{eq:control.2a}--\eqref{eq:control.2d}.\footnote{For each (directed)
link $(i,j)\in E$ we assume that only node $i$ 
maintains the variables $(\eta^+_{ij}(t), \eta^-_{ji}(t), \tilde\phi_{ij}(t) )$.
In practice, node $j$ will probably maintain symmetric variables to reduce 
communication burden or for other reasons outside our mathematical model here.}
In contrast to the \emph{completely decentralized} control derived in the 
per-node balance case,  here the control is \emph{distributed} where each 
node $i$ updates $(\lambda_i(t), \eta^+_{ij}(t), \eta^-_{ij}(t) )$ using only local 
measurements or computation but requires the information $(\lambda_j(t), z_j(t))$
from its neighbors $j$ to update $\tilde\phi_{ij}(t)$.
Note that $z_i(t)$ is not a variable but a shorthand for (function)
$P^g_j(t)-P^l_j(t)-p_j-\hat{U}_j(\tilde\phi(t))$.
The control inputs $(u^g_i(t), u^l_i(t) )$ 
in \eqref{eq:control.2e}\eqref{eq:control.2f} are
functions of the network state $(P^g_i(t), P^l_i(t), \omega_i(t) )$ and the
internal state $(\lambda_i(t), \eta^+_{ij}(t), \eta^-_{ij}(t), \tilde\phi_{ij}(t) )$. 
We write $u^g_j$ and $u^l_j$ as functions of 
$(P^g_j, P^l_j, \omega_j, \lambda_j)$: for $j\in N$
\begin{subequations}
        \bq
        u^g_j(t) & := & u^g_j \left( P^g_j(t), \omega_j(t), \lambda_j(t), z_j(t)\right)
        \label{eq:control.2a'}
        \\
        u^l_j(t) & := & u^l_j \left( P^l_j(t), \omega_j(t), \lambda_j(t), z_j(t) \right)
        \label{eq:control.2b'}
        \eq
        where the functions are defined by the right-hand sides of
        \eqref{eq:control.2e}\eqref{eq:control.2f}.
        \label{eq:control.2'}
\end{subequations}

Now we comment on the implementation of the control \eqref{eq:control.2}. 

\begin{remark}[Implementation]
\bee
         \item As discussed above, communication is needed only between neighboring
         nodes (areas) to update the variables $\tilde\phi_{ij}(t)$.
%        The internal variables $(\lambda_j(t), \phi_{ij}, \eta^+_{ij}, \eta^-_{ij})$ in \eqref{eq:control.2a}--\eqref{eq:control.2d} are cyber quantities that are
%        computed at each node $j$, where  $(\lambda_j(t),  \eta^+_{ij}, \eta^-_{ij})$ can be computed at node $j$ based on local information while $\phi_{ij}$ requires to know $\lambda_i(t)$ and $z_i(t)$ of its neighboring nodes $i\in N$. 
%        We can also directly compute $z_i(t)$ by differentiating $\lambda_i(t)$, hence further reducing the communication requirement. 
%        Furthermore, according to the definition of $z_i$, we have $\dot {\lambda}_i= \gamma^{\lambda}_i z_i$. Hence, if the constant $\gamma^{\lambda}_i$ is known to node $j$ in advance, $z_i$ can be calculated from $\lambda_i$, with no need to exchange between neighbouring nodes. \slow{But does node $j$ needs $z_i$ to update $\phi_{ij}$, and 
%        if so, how does $j$ get it without communication?}
%    \fliu{I think we made a mistake, since we do need the information of $\phi_i$. So the controller should require to exchange the information  of both $\lambda_i$ and $\phi_i$. Zhaojian, I remember we had discussed this issue before? }
%       \wang{I agree. both $\lambda_i$ and $\phi_i$ should be exchanged. In addition, I think "the requirement to communication can be minimized" should be removed, as there is no definition about what is the minimal communication.} 
        \item
        Similar to the per-node power balance case, we can avoid measuring the load change $p_j$  
        by using \eqref{eq:model.1b} and the definition of $z_j(t)$ to replace \eqref{eq:control.2a} with
        \bqn
        z_j(t) & = & 
         M_j \dot \omega_j + D_j\omega_j(t) - \sum_{i: i\rightarrow j} \! P_{ij}(t)
        +  \sum_{k: j\rightarrow k} \! P_{jk}(t) 
        \nonumber \\
          && \ - \ \hat{U}_j(\tilde\phi(t))
           \nonumber\\
        \dot\lambda_j &=& \gamma^{\lambda}_j \, z_j(t)
        \eqn
\eee
\end{remark}

\newcounter{TempEqCnt}
\setcounter{TempEqCnt}{\value{equation}}
\setcounter{equation}{9}
\begin{figure*}[!t]
	
	%	\begin{equation}
	
	\bq
	L_2(x; \rho) & \!\!\!\! = \!\!\!\! &  
	\frac{1}{2} \left( \sum_{j\in N} \alpha_j \left(P^g_j\right)^2 + \sum_{j\in N} \beta_j \left(P^l_j\right)^2
	+ \sum_{j\in N}  D_j \omega_j^2+ \sum_j z_j^2 \right)
	%	\nonumber\\
	%	&&
	\ + \ \sum_{j\in N} \mu_j \left(P^g_j - P^l_j  - p_j -  D_j \omega_j 
	+ \sum_{i: i\rightarrow j}  B_{ij} \theta_{ij} -  \sum_{k: j\rightarrow k} \! B_{jk}\theta_{jk}  \!\right)
	\nonumber\\ 
	& & 
	+\ \sum_{j\in N} \lambda_j \left(
	P^g_j - P^l_j  - p_j  
	+ \sum_{i: i\rightarrow j} B_{ij} \phi_{ij}
	- \sum_{k: j\rightarrow k} B_{jk} \phi_{jk} \!\right) 
	%	\nonumber\\
	%	&&
	\ + \sum_{(i,j)\in E} \eta^-_{ij} \left(\underline{\theta}_{ij}-\phi_{ij}(t)\right)
	\ + \sum_{(i,j)\in E} \eta^+_{ij} \left(\phi_{ij}(t)-\overline{\theta}_{ij}\right)
	\label{eq:defL.2}
	\eq
%\slow{Should $L_2$ in \eqref{eq:defL.2} be $\cdots 
%\sum_{j\in N} \lambda_j \left(
%	P^g_j - P^l_j  - p_j  
%	+ \sum_{i: i\rightarrow j} B_{ij} \phi_{ij}
%	- \sum_{k: j\rightarrow k} B_{jk} \phi_{jk} \!\right) \cdots$?}
	
	%	\end{equation}
	%\setcounter{equation}{\value{mytempeqncnt}}
	
	%%%%%%%%%
	\hrulefill
	\vspace*{2pt}
\end{figure*}  
\setcounter{equation}{\value{TempEqCnt}} 

\subsection{Design rationale}
\label{subsec:design2}
The controller design \eqref{eq:control.2} is also motivated by a (partial) primal-dual algorithm for
\eqref{eq:opt.2}, as for the per-node power balance case.   
% We first review the form of a standard primal-dual algorithm and then explain that {the closed-loop dynamics
%	\eqref{eq:model.1}\eqref{eq:control.2} carry out such an algorithm for 
%	\eqref{eq:opt.2} in real time over the closed-loop system.} 

\vspace{0.1in}
\noindent
\textbf{Primal-dual algorithms.}
The optimization problem in the network balance case differs from that in the per-node balance case
in the inequalities \eqref{eq:opt.2c} on $\tilde\phi$.
Consider a general constrained convex optimization with inequality constraints:
\bqn
\min_{x\in X} \ \ f(x) & s. t. & g(x) = 0, \ \ h(x)\le 0
\eqn
where $f:\mathbb R^n\rightarrow \mathbb R$, $g:\mathbb R^n\rightarrow \mathbb R^{k_1}$, $h:\mathbb R^n\rightarrow \mathbb R^{k_2}$,
and $X\subseteq \mathbb R^n$ is closed and convex. 
Here an inequality constraint $h(x)\le 0$ is imposed explicitly. 
Let $\rho_1\in\mathbb R^{k_1}$ be the Lagrange multiplier associated
with the equality constraint $g(x)=0$, $\rho_2\in\mathbb R^{k_2}$ that associated with the inequality 
constraint $h(x)\le 0$, and $\rho := (\rho_1, \rho_2)$.  
Define the Lagrangian $L(x; \rho) := f(x) + \rho_1^T g(x)+ \rho_2^T h(x)$.
A standard primal-dual algorithm takes the form:
%The controller design \eqref{eq:control.2} is also motivated by a partial primal-dual algorithm for
%\eqref{eq:opt.2} that dualizes the constraints \eqref{eq:lineConstraint} and \eqref{eq:opt.2b}, 
%as we  explain. 
%The corresponding Lagrangian of \eqref{eq:opt.2} is defined by \eqref{eq:defL.2}. Denote $x:=(P^g, p^l, \omega, \theta,  \phi)$ as the primal variables, 
%$\rho_1 := (\lambda, \mu )$, $\rho_2 := (\eta^+, \eta^- )$ and $\rho := (\rho_1, \rho_2 )$ are Lagrange multipliers of \eqref{eq:opt.2}.
%% and $(\underline\gamma, \overline\gamma, \underline\nu, \overline\nu)$ are nonnegative.  
%
%Similar to (\ref{eq:pma.1}), a primal-dual algorithm with inequality constraints takes the form
%\begin{subequations}
%        \bq
%        \dot x & = & - \Gamma_x\, \frac{\partial L_2}{\partial x} \left(x(t), \rho(t) \right)
%        \label{eq:primaldual.2a}
%        \\
%        \dot \rho & = &   \Gamma_\rho\, \left [ \frac{\partial L_2}{\partial \rho}(x(t), \rho(t)) \right]^+_{\rho}
%        \label{eq:primaldual.2b}
%        \eq
%        \label{eq:primaldual.2}
%\end{subequations}
\begin{subequations}
	\bq
	x(t+1) & := & \text{Proj}_X \left( x(t) \ - \ \Gamma^x\, \nabla_x L(x(t); \rho(t)) \right)
\label{eq:primaldual.2a}
	\\
	\rho_1(t+1) & := & \rho_1(t) \ + \ \Gamma^{\rho_1}\, \nabla_{\rho_1} L(x(t); \rho(t))
	\label{eq:primaldual.2b}
\\
	\rho_2(t+1) & := & \left(\rho_2(t) \ + \  \Gamma^{\rho_2}\, \nabla_{\rho_2} L(x(t); \rho(t))\right)^+
\label{eq:primaldual.2c}
	\eq
where $\Gamma^x, \Gamma^{\rho_1}, \Gamma^{\rho_2}$ are strictly positive diagonal gain matrices.
Here, if $a$ is a scalar then $(a)^+ := \max\{a, 0\}$ and if $a$ is a vector then 
$(a)^+$ is defined accordingly componentwise.
For a dual algorithm, \eqref{eq:primaldual.2a} is replaced by 
	\bq
	x(t) & := & \min_{x\in X}\, L(x; \rho(t))
	\label{eq:da.2b}
	\eq
As for the per-node balance case, all variables in $x(t)$ are updated according to
\eqref{eq:primaldual.2a} except $\omega(t)$ which is updated according to \eqref{eq:da.2b},
as we see below.
\label{eq:primaldual.2}
\end{subequations}

The set $X$ in \eqref{eq:primaldual.2a} is defined by the constraints 
\eqref{eq:OpConstraints.1}:
\bq
\!\!\!\!\!
X & \!\!\!\!\ := \!\!\!\! & \left\{(P^g, P^l) : 
(\underline{P}^g, \underline{P}^l)  \ \leq \ (P^g, P^l) \ \leq \
(\overline{P}^g, \overline{P}^l)  \right\}
\label{eq:defX}
\eq

\vspace{0.1in}
\noindent
\textbf{Controller \eqref{eq:control.2} design.} Let $\rho_1 := (\lambda, \mu)$ be
 the Lagrange multipliers associated with constraints \eqref{eq:opt.2b} and \eqref{eq:opt.2a}
 respectively,
 $\rho_2 := (\eta^+, \eta^-)$ the  multipliers associated with constraints \eqref{eq:opt.2c}, 
 and $\rho:=(\rho_1, \rho_2)$.
Define the Lagrangian of \eqref{eq:opt.2} by (\ref{eq:defL.2}). 
Note that it is only a function of $(x, \rho)$ and independent of $u := (u^g, u^l)$ as we 
treat $u$ as a function of $(x, \rho)$
defined by the right-hand sides of \eqref{eq:control.2e}\eqref{eq:control.2f}.
%The set $X$ in \eqref{eq:primaldual.2a} or \eqref{eq:da.2b} is the same as (\ref{eq:defX}).

The closed-loop dynamics \eqref{eq:model.1}\eqref{eq:control.2} carry out 
        an approximate primal-dual algorithm \eqref{eq:primaldual.2} for 
        solving \eqref{eq:opt.2} in real time over the coupled physical power network and cyber computation. 
Since the reasoning is similar to the per-node balance case, we only provide a summary.
Rewrite the Lagrangian $L_2$ in  vector form 
\setcounter{equation}{10} 
\bq
\label{eq:defL.21}
L_2(x; \rho) &  =& \frac{1}{2} \left( (P^g)^T A^g P^g + (P^l)^T A^l P^l + \omega^T D \omega + z^Tz \right)
\nonumber \\
&& \ + \ \lambda^T \!\! \left( P^g - P^l - p - CB\tilde\phi \right) 
\nonumber \\
& & \ + \ \mu^T \!\! \left( P^g - P^l - p -D\omega - C B\tilde \theta \right)
\nonumber \\
&& \ + \ (\eta^+)^T \left(\tilde \phi-\overline{\theta} \right) \ + \ (\eta^-)^T\left(\underline{\theta}-\tilde\phi\right)
\eq
where $A^l := \diag(\beta_j, j\in N)$, $B:=\diag(B_{ij}, (i,j)\in E)$.
%\slow{Should $L_2$ be
%\bq
%\label{eq:defL.21}
%L_2&  =& \frac{1}{2} \left( (P^g)^T A^g P^g + (P^l)^T A^l P^l + \omega^T D \omega + z^Tz \right)
%\nonumber \\
% && \ + \ \lambda^T \!\! \left( P^g - P^l - p - CB\tilde\phi \right) 
% \nonumber \\
% & & \ + \ \mu^T \!\! \left( P^g - P^l - p -D\omega - C B\tilde \theta \right)
%\nonumber \\
%&& \ + \ (\eta^+)^T \left(\tilde \phi-\overline{\theta} \right) \ + \ (\eta^-)^T\left(\underline{\theta}-\tilde\phi\right)
%\eq
%}

First, the control \eqref{eq:control.2b}\eqref{eq:control.2c} can be interpreted as
a continuous-time version of the dual update \eqref{eq:primaldual.2c}
on the dual variable $\rho_2 := (\eta^+(t), \eta^-(t))$:
\begin{subequations}
\bqn
\dot \eta^+ & = & \Gamma^{\eta} \, \left[ \nabla_{\eta^+} L_2 (x(t); \rho(t)) \right]^+_{\eta^+(t)}
\\
\dot \eta^- & = & \Gamma^{\eta^-} \, \left[ \nabla_{\eta^-} L_2 (x(t); \rho(t)) \right]^+_{\eta^-}
\eqn
where $\Gamma^{\eta} := \diag(\gamma^{\eta}_{ij}, (i,j)\in E)$.

Second, the control \eqref{eq:control.2a} carries out the dual update 
\eqref{eq:primaldual.2b} on $\lambda(t)$:  
        \bq
        \dot\lambda & = &  \Gamma^{\lambda}\ \nabla_{\lambda} L_2 (x(t), \rho(t))
        \label{eq:dual.2d}
        \eq
where $\Gamma^{\lambda} := \diag(\gamma^{\lambda}_j, j\in N)$.  
The swing dynamic \eqref{eq:model.1b} carries out  
the dual update \eqref{eq:primaldual.2b} on $\mu(t)$ because, 
as in the per-node balance case, we can identify $\mu(t) \equiv \omega(t)$ so that
        \bq
        \dot\mu & = & \dot\omega \ \ = \ \ M^{-1}\ \nabla_\mu L_2 (x(t); \rho(t))
        \label{eq:dual.2c}
        \eq
where $M := \diag(M_j, j\in N)$.  

Finally we  show that \eqref{eq:model.1a},
\eqref{eq:model.1c}, \eqref{eq:model.1d},  and \eqref{eq:control.2d} implement a
mix of the primal updates \eqref{eq:primaldual.2a} and \eqref{eq:da.2b} on the primal variables
$x := (\tilde{\theta}(t ); \tilde{\phi}(t ); \omega(t ); P^g(t ); P^l(t ))$.
Setting $\omega(t) \equiv \mu(t)$ is equivalent to the primal update \eqref{eq:da.2b}
on $\omega(t)$, as in the per-node balance case.   Moreover the
control laws \eqref{eq:control.2e}\eqref{eq:control.2f} are then equivalent to 
\begin{align}
        T^g\dot P^g & \ =\  \left[ P^g(t) - \, \Gamma^g\, \nabla_{P^g} L_2 (x(t), \rho(t))
        \right]_{\underline P^g}^{\overline P^g} \ - \ P^g(t)
        \label{eq:primal.2a} \\    
        T^l\dot P^l & \ =\  \left[ P^l(t) - \, \Gamma^l\, \nabla_{P^l} L_2 (x(t), \rho(t))
        \right]_{\underline P^l}^{\overline P^l} \ - \ P^l(t)    
        \label{eq:primal.2b} 
\end{align}
i.e., the generator and controllable load at each node $j$ carry out the primal update 
\eqref{eq:primaldual.2a}.  
For $(\tilde\theta, \tilde\phi)$, \eqref{eq:model.1a} and \eqref{eq:control.2d} are equivalent
to the primal update \eqref{eq:primaldual.2a}:
        \bq
        \dot{ \tilde{\theta}}& = & - B^{-1}\nabla_{\tilde\theta} L_2 (x(t), \rho(t))
        \label{eq:primal.2c}
\\
        \dot {\tilde{\phi}} & = & -\Gamma^{\tilde{\phi}} \ \nabla_{\tilde{\phi}} L_2 (x(t), \rho(t))
        \label{eq:primal.2d}
        \eq
        where $\Gamma^{\tilde{\phi}} := \diag(\gamma^{\tilde\phi}_{ij}, (i,j)\in E)$
        \label{eq:primaldual.3}. 
\end{subequations}

\subsection{Optimality of equilibrium point}
\label{subsec:optimality.1}

In this subsection, we address the optimality of the equilibrium point of the closed-loop system \eqref{eq:model.1}\eqref{eq:control.2}. 
Given an $(x, \rho) := \left( (\tilde{\theta}, {\tilde{\phi}}, \omega,  P^{g}, P^{l}), (\lambda, \mu), (\eta^-, \eta^+)\right)$, recall that
the control input $u(x, \rho_1, \rho_2)$ is given by \eqref{eq:control.2'}.
\begin{definition}
        \label{def:ep.2}
        A point $(x^*, \rho^*) := ( \tilde{\theta}^*, {\tilde \phi}^*, \omega^*,  P^{g*}, P^{l*}, \lambda^*, \eta^{+*}, $ $ \eta^{-*}, \mu^* )$
        is an \emph{equilibrium point} or an \emph{equilibrium} of the closed-loop system 
        \eqref{eq:model.1}\eqref{eq:control.2} if 
        \bee
        \item The right-hand side of \eqref{eq:model.1} vanishes at $x^*$ and $u(x^*, \rho^*)$.  
        \item The right-hand side of \eqref{eq:control.2a}--\eqref{eq:control.2d} vanishes at $(x^*, \rho^*)$.
        \eee
\end{definition}

\begin{definition}
        A point $(x^*, \rho^*)$ is \emph{primal-dual optimal} if $(x^*, u(x^*, \rho^*))$ is optimal 
        for \eqref{eq:opt.2} and $\rho^*$ is optimal for its dual problem.
\end{definition}

We  make the following assumption:
\bi
\item[\textbf{A2:}]  
     The problem (\ref{eq:opt.2}) is feasible. 
\ei 
 
 The following theorem characterizes the correspondence between the equilibrium of the closed-loop system \eqref{eq:model.1}\eqref{eq:control.2} and the primal-dual optimal solution of  \eqref{eq:opt.2}. 
\begin{theorem}
        \label{thm:5}
        Suppose  A2  holds.   A point  $(x^*, \rho^*)$ is primal-dual optimal if
        and only if  $(x^*, \rho^*)$ is an equilibrium of  closed-loop 
        system \eqref{eq:model.1}\eqref{eq:control.2}  satisfying \eqref{eq:OpConstraints.1}
        and $\mu^*=0$.
\end{theorem}

Next result says that, at equilibrium, the network balance condition \eqref{eq:balance.net} and 
line limits \eqref{eq:lineConstraint} are satisfied and the nominal frequency is
restored.  Moreover the equilibrium is unique.
\begin{theorem}
        \label{thm:6}
        Suppose  A1 and A2 hold.  Let $(x^*, \rho^*)$ be primal-dual optimal. 
        Then
        \bee
%        \item for $(i,j)\in E$, if $\underline \theta _{ij}<\phi^* _{ij}<\overline \theta_{ij}$, $\lambda_i^*=\lambda_j^*$.
        \item The equilibrium
        $(x^*, \mu^*)$ is unique, with $(\theta^*, \phi^*)$ being unique
        up to (equilibrium) reference angles $(\theta_0, \phi_0)$.
%        \item
%        $\rho^*$ is unique with the condition that $\sum\nolimits_j(\underline{P}^g_j - \overline{P}^l_j )<\sum\nolimits_j p_j$,  $\sum\nolimits_j(\overline{P}^g_j - \underline{P}^l_j ) > \sum\nolimits_j p_j $, $\forall j\in N$ and in the steady state, for $\forall j\in N$, the generation ($\underline P_j^g, \overline P_j^g$) and controllable load constraints ($\underline P_j^l, \overline P_j^l$)  cannot be reached  with   tie lines powers connected to area $j$ ($\underline P_{ij}, \overline P_{ij}$, $i\in N_j$)  reaching their limits simultaneously.        
        % \item $x^*$ achieves per-node power balance \eqref{eq:balance.node} and
        %               satisfies the operational constraints \eqref{eq:OpConstraints.1};
        \item The nominal frequency is restored, i.e., $\omega^*_j=0$ for all $j\in N$;
        moreover $\tilde\phi^*_{ij} = \tilde\theta^*_{ij}$ for all $(i,j)\in E$.
        \item The network balance condition \eqref{eq:balance.net} is satisfied by $x^*$.
        \item The line limits \eqref{eq:lineConstraint} are satisfied by $x^*$, implying
        $\underline{P}_{ij} \le P_{ij} \le \overline{P}_{ij}$ on every tie line $(i,j)\in E$.
%        \item If $\underline{P}^g_j<P^{g*}_j<\overline{P}^g_j$  and
%        $\underline{P}^l_j<P^{l*}_j<\overline{P}^l_j$ then, at optimality, the marginal 
%        generation regulation cost is equal to the marginal load regulation cost at
%        node $j$, i.e., $\alpha_j P^{g*}_j = -\beta_j P^{l*}_j = -\lambda_j^*$.
        \eee 
\end{theorem}
Theorem \ref{thm:6} shows that the equilibrium point has a simple yet intuitive structure. 
Moreover, Theorem \ref{thm:6} implies that the closed-loop system can autonomously 
eliminate congestions on tie lines. This feature has important implications. It means our distributed frequency control is capable of serving as a corrective re-dispatch without the coordination of dispatch centers if a congestion 
arises.   This can enlarge the feasible region for economic dispatch,
since corrective re-dispatch has been naturally taken into account. 

The proofs of Theorem \ref{thm:5} and \ref{thm:6} are given in Appendix A.

\subsection{Asymptotic stability}
 In this subsection, we address the asymptotic stability of the closed-loop 
 system \eqref{eq:model.1}\eqref{eq:control.2}, under an additional assumption:
 \bi
 \item[\textbf{A3:}] The initial state of the closed-loop system  \eqref{eq:model.1}\eqref{eq:control.2} is finite, and $p^g_j(0)$, $p^l_j(0)$  satisfy constraint \eqref{eq:OpConstraints.1}. 
 \ei
 
 As in the per-node balance case the closed-loop system \eqref{eq:model.1}\eqref{eq:control.2} 
satisfies constraint \eqref{eq:OpConstraints.1} even during transient.
\begin{lemma}
        \label{lemma:bounded.2}
        Suppose A1 and A3 hold. 
        Then constraint \eqref{eq:OpConstraints.1} is satisfied for all  $t>0$, i.e. 
        $(P^g(t), P^l(t))\in X$ for all $t\geq 0$ where $X$ is defined in \eqref{eq:defX}.    
\end{lemma}
The proof is exactly the same as that for Lemma 3 in \cite{Wang:DistributedFrequency} and omitted.

%As for the network balance case, we set $\mu(t)\equiv \omega(t)$ and use $\tilde\theta, \tilde\phi$
%instead of $\theta,\phi$ since there is a bijection between these variables once we fix 
%$\theta_0(t) \equiv \phi_0(t) \equiv 0$ as the reference.   We denote the vector of state variables 
%by $w:=(\tilde {\theta}, \omega, P^g, P^l, \lambda, \eta^+, \eta^-, \tilde{\phi})$.
% 

Similar to the per-node balance case, we first rewrite the closed-loop system using states 
$\tilde{\theta}, \tilde{\phi}$ instead of $\theta, \phi$ (they are equivalent under assumption A1).
Setting $\mu\equiv \omega$, the closed-loop system \eqref{eq:model.1}\eqref{eq:control.2} 
is equivalent to (in vector form):
        \begin{subequations}
\bq
                \dot {\tilde \theta}(t)&=&   C^T\omega(t) \\
\label{eq:closedloop.2a}
                \nonumber
                \dot \omega (t)&=&M^{-1}\left(P^g(t)-P^l(t)-p-D\omega(t)-CB\tilde \theta(t) \right ) \label{eq:closedloop.2b}\\
                \label{eq:model.2b}\\
                \dot{P}^g(t)&=&(T^g)^{-1}\left ( -P^g(t)+\hat u^g(t) \right)
\label{eq:closedloop.2c}\\ 
                \dot{P}^l(t)&=&(T^l)^{-1}\left ( -P^l(t)+\hat u^l(t) \right )
\label{eq:closedloop.2d}\\
\dot \eta^+(t)&=& \Gamma^{\eta}[\tilde\phi(t)-\overline \theta]^+_{\eta^+}
\label{eq:closedloop.2f}\\
\dot \eta^-(t)&=& \Gamma^{\eta}[\underline \theta-\tilde\phi(t)]^+_{\eta^-}
\label{eq:closedloop.2g}\\
                \dot{\lambda}(t)&=&\Gamma^{\lambda} \left (P^g(t)-P^l(t)-p-CB\tilde \phi(t) \right)
\label{eq:closedloop.2e}\\
\dot {\tilde \phi}(t)&=& \Gamma^{\tilde \phi} \left ( BC^T\lambda(t)
+BC^Tz(t)+\eta^-(t)-\eta^+(t)
\label{eq:closedloop.2h}
\right) 
\eq
where $z(t) := P^g(t) - P^l(t) - p - CB\tilde \phi(t)$ and
\label{eq:closedloop.2}
\end{subequations}
\bqn
	\hat u^g(t) & = & \left[ 
	P^g(t) - \Gamma^g \left( A^g P^g_j(t) + \omega(t)+z(t) + \lambda(t) \right) \right]_{\underline P^g}^{\overline P^g}
	\\
	\hat u^l_j(t) & = & \left[ 
	P^l(t) - \Gamma^l \left( A^l P^l_j(t) - \omega(t)-z(t) - \lambda(t)\right)
	\right]_{\underline P^l}^{\overline P^l}
\eqn
%where, $A^g:=\diag\{\alpha_j, j\in N\}$ and $A^l:=\diag\{\beta_j, j\in N\}$.
Denote $w:=(\tilde {\theta}, \omega, P^g, P^l, \lambda, \eta^+, \eta^-, \tilde{\phi})$.

Note that the right-hand sides of \eqref{eq:closedloop.2f}\eqref{eq:closedloop.2g}
are discontinuous due to projection to the nonnegative quadrant for 
$(\eta^{+}(t), \eta^-(t))$.  The system
\eqref{eq:closedloop.2} is called a projected dynamical system and we adopt the
concept of Caratheodory solutions for such a system where a trajectory $(w(t), t\geq 0)$
is called a Caratheodory solution, or just a solution, to \eqref{eq:closedloop.2} if it is 
absolutely continuous in $t$ and satisfies \eqref{eq:closedloop.2} almost everywhere.
The result in \cite[Theorems 2 and 3]{DupuisNagurney1993} implies that, given any initial
state, there exists a unique solution trajectory to the closed-loop 
system \eqref{eq:closedloop.2} as the unprojected system is Lipschitz and
the nonnegative quadrant is closed and convex.
See \cite[Theorem 3.1]{Monica:Existence} for extension of this result to 
the Hilbert space.

With regard to system \eqref{eq:closedloop.2}, we first define two sets, $\sigma^+$ and $\sigma^-$, as follows \cite{Feijer:Stability}.
\bqn
\sigma^+ &:=& \{(i,j)\in E \, | \, \eta^+_{ij}=0,         \, \tilde{\phi}_{ij}-\overline{\theta}_{ij}<0
\}\\
\sigma^- &:=& \{(i,j)\in E \, | \, \eta^-_{ij}=0,         \, \underline{\theta}_{ij}-\tilde{\phi}_{ij}<0
\}
\eqn
Then \eqref{eq:control.2b} and \eqref{eq:control.2c}
are equivalent to
\begin{subequations}
	\bq
	\dot\eta^+_{ij} &=& \left\{
	\begin{array}{ll}
		\gamma^{\eta}_{ij}(\tilde{\phi}_{ij}-\overline{\theta}_{ij}),
		& \text{if} (i,j) \notin \sigma^+ ;\\
		0,
		& \text{if} (i,j) \in \sigma^+ .        \end{array}
	\right.\\
	\dot\eta^-_{ij} &=& \left\{
	\begin{array}{ll}
		\gamma^{\eta}_{ij}(\underline{\theta}_{ij}-\tilde{\phi}_{ij}),
		& \text{if} (i,j) \notin \sigma^- ;\\
		0,
		& \text{if} (i,j) \in \sigma^- .        \end{array}
	\right.
	\eq
	\label{eq:eta}
\end{subequations}

In a fixed $\sigma^+, \sigma^-$, define $F(w)$. 
% \slow{I think $F$ must be continuously differentiable, so cannot include projection 
% of $\eta^+, \eta^-$.} 
 \bq
  F(w)&=&\left [ 
   \begin{array}{l}
               -B^{1/2}C^T\omega \\
               -M^{-1/2}\left(P^g-P^l-p-D\omega-CB\tilde \theta \right )\\ 
               (T^g)^{-1} \left(A^g P^g+\omega+z+\lambda\right)\\
               (T^l)^{-1} \left( A^l P^l-\omega-z-\lambda\right)\\ 
               -(\Gamma^{\eta})^{1/2}[\tilde\phi-\overline \theta]^+_{\eta^+} \\
               -(\Gamma^{\eta})^{1/2}[\underline \theta-\tilde\phi]^+_{\eta^-} \\
               -(\Gamma^{\lambda})^{1/2} \left (P^g-P^l-p-CB\tilde
\phi \right)\\               
              -(\Gamma^{\tilde{\phi}})^{1/2} \left ( BC^T\lambda+BC^Tz+\eta^- - \eta^+ \right)
    \end{array}  \right ] 
         \label{eq:Fz.2}
 \eq
%        where $\Gamma ^{\tilde \theta}:=\diag(\sqrt {B^{-1}_{ij}}, (i,j)\in E)$; $\Gamma ^{\omega}:=\diag(\sqrt {M^{-1}_j}, j\in N)$; $\Gamma ^{g}:=\diag((T^g_j)^{-1}, j\in N)$; $\Gamma ^{l}:=\diag((T^l_j)^{-1}, j\in N)$; $\Gamma ^{\eta}:=\diag(\sqrt{\gamma^{\eta}_{ij}}, (i,j)\in E)$; $\Gamma ^{\lambda}:=\diag(\sqrt{\gamma^{\lambda}_j}, j\in N)$; $\Gamma ^{\tilde \phi}:=\diag(\sqrt{\gamma^{\tilde
%\phi}_j}, j\in N)$. 
If $\sigma^+$ and $\sigma^-$ do not change, $F(w)$ is continuously differentiable in $w$.
%\slow{Let's use the same notation as much as possible between per-node and
%network balance cases.  We can either change the notations in per-node case
%or in network case, e.g., $\Gamma$'s.}
        
        %Note that $\tilde \theta(t):=C^T\theta(t)$. Then we still have
        %\bqn
        %\dot {\tilde \theta}& =& -\Gamma^{\tilde \theta} \nabla_{\tilde \theta}L
        %\eqn
        
       Similarly, we define 
        $
        S   :=  \mathbb R^{m+n+1}\times X \times \mathbb R^{2m+n+1+m}
        $, where the closed convex set $X$ is defined in \eqref{eq:defX}.
%        \slow{Can/should we re-order \eqref{eq:closedloop.2e} in \eqref{eq:closedloop.2} and
%        everything following that (e.g., definition of $F(w)$), so that
%        $w:=(\tilde {\theta}, \omega, P^g, P^l, \eta^+, \eta^-, \lambda, \tilde{\phi})$?
%        If not, we will re-order $S$ to be 
%        $S   :=  \mathbb R^{m+n+1}\times X \times \mathbb R^{n+1} \times \mathbb R^{2m}_+ \times \mathbb R^{m}$.
%        }
         Then for any $w$ we define the projection of $w-F(w)$ onto  $S$ as  
$$H(w) :=  \text{Proj}_S (w-F(w)) := \ \arg\min_{y\in S} \| y - (w-F(w)) \|_2 $$
Then the closed-loop system \eqref{eq:closedloop.2} is equivalent to 
\bq
\label{eq:model.4}
\dot w(t)&=&\Gamma_2 (H(w(t))-w(t))
\eq
where the positive definite gain matrix is :
\begin{align}
	\Gamma_2\  =\ \diag\ \bigg(B^{-1/2}, M^{-1/2}, (T^g)^{-1},  (T^l)^{-1},\nonumber\\
	 (\Gamma^{\lambda})^{1/2}, (\Gamma^{\eta})^{1/2}, (\Gamma^{\eta})^{1/2}, (\Gamma^{\tilde \phi})^{1/2} \bigg) \nonumber
\end{align}

Note that the projection operation $H$ has an effect only on $(P^g; P^l)$ and Lemma \ref{lemma:bounded.2} indicates that $w(t )\in  S$ for all $t>0$, justifying the equivalence of \eqref{eq:closedloop.2} and (\ref{eq:model.4}).

A point $w^*\in S$ is an equilibrium of the closed-loop system (\ref{eq:model.4}) if and only if it is a fixed point of the projection
$H(w^*)=w^*$. 
Let $E_2:=\{\ w\in S\ |\ H(w(t))-w(t)=0\ \}$ be the set of equilibrium points. Then we have the following theorem.

\begin{theorem}
	\label{thm:stability.22}
	Suppose A1, A2 and A3 hold.
	Starting from any initial point $w(0)$, 
	$w(t)$ remains in a bounded set for all $t$ 
	and $w(t)\rightarrow w^*$ as $t\rightarrow\infty$ for some equilibrium $w^*\in E_2$ that
	is optimal for problem (\ref{eq:opt.2}).
\end{theorem}

For any equilibrium point $w^*$, we define the following  function taking the same form as the per-node case.
\begin{align}
	\label{eq:lyapunov.2}
	\tilde{V}_2(w)&=-(H(w)-w)^TF(w)
	-\frac{1}{2} ||H(w)-w||^2_2 \nonumber\\
	&\quad +\frac{1}{2}k(w-w^*)^T\Gamma_2^{-2}(w-w^*) 
\end{align}
where   $k$ is small enough such that $\Gamma_2-k\Gamma_2^{-1} > 0$ is strictly positive definite. 

For any fixed $\sigma^+$ and $\sigma^-$, $\tilde{V}_2$ is continuously differentiable as $F(w)$ is continuously differentiable in this situation. Similar to $V_1(w)$ used in Part I of the paper, we know $ \tilde{V}_2(w) \geq 0$ on $S$ and $ \tilde{V}_2(w)=0$ holds only at any equilibrium $w^* = H(w^*)$\cite{Fukushima:Equivalent}. Moreover, $\tilde{V}_2$ is nonincreasing for fixed $\sigma^+$ and $ \sigma^-$, as we prove in  Appendix \ref{appd.thm8}. 

It is worth to note that the index sets $\sigma^+$ and $\sigma^-$ may change sometimes, resulting in discontinuity of $\tilde{V}_2(w)$. To circumvent such an issue, we slightly modify the definition of $V_2(w)$ at the discontinuous points as:
\begin{enumerate}
	\item $V_2(w) := \tilde{V}_2(w)$, if $\tilde{V}_2(w)$ is continuous at $w$;
	\item $V_2(w) := \limsup\limits_{v\to w} \tilde{V}_2(v)$, if  $\tilde{V}_2(w)$ is discontinuous at $w$.
\end{enumerate}
Then $V_2(w)$ is upper semi-continuous in $w$, and  $ V_2(w) \geq 0$ on $S$ and $ V_2(w)=0$ holds only at any equilibrium $w^* = H(w^*)$. 
As  $V_2(w)$ is not differentiable for $w$ at discontinuous points, we use the Clarke gradient as the gradient at these points \cite{clarke:optimization}.

%further define the gradient of $V_2(w)$ as follows. 
%\begin{enumerate}
%	\item $\frac{\partial V_2}{\partial w}(w):= \frac{\partial\tilde V_2}{\partial w}(w)$, if $\tilde{V}_2(w)$ is continuous at $w$;
%	\item $\frac{\partial V_2}{\partial w}(w) := \limsup\limits_{v\to w} \left(\frac{ \tilde{V}_2(v) - \tilde{V}_2(w)}{v-w}\right)$, if  $\tilde{V}_2(w)$ is discontinuous at $w$.
%\end{enumerate}
%\wang{I think 2) should be $\frac{\partial V_2}{\partial w}(w) := \lim\limits_{v\to w} \bigg(\frac{ \limsup\limits_{v\to w}\tilde{V}_2(v) - \tilde{V}_2(w)}{v-w}\bigg)$, but it is so strange.}
%
%Here, we just use $\frac{\partial V_2}{\partial w}(w(t_k^-))$ as its Clarke gradient at $w(t_k)$ \cite{clarke:optimization}.

Note that $\tilde{V}_2$ is continuous almost everywhere except the switching points. Hence both $V_2(w)$ is \emph{nonpathological} \cite{Bacciotti:Nonpathological,bacciotti:stability}. With these definitions and notations above, we can prove  Theorem \ref{thm:stability.22}. The detail of proof is provided in  Appendix \ref{appd.thm8}.

\section{Case studies}
\subsection{System configuration}
A four-area system based on Kundur's four-machine, two-area system \cite{Fang:Design} \cite{Kundur:Power} is used to test our optimal frequency controller. There are one (aggregate) generator (Gen1$\sim$Gen4), one controllable (aggregate) load (L1c$\sim$L4c) and one uncontrollable (aggregate) load (L1$\sim$L4) in each area, which is shown in Fig.\ref{fig:system}. The parameters of generators and controllable loads are given in Table \ref{tab:SysPara}. The total uncontrollable load in each area are identically 480MW. At time $t=10s$, we add step changes on the uncontrollable loads in four areas to test the performance of our controllers. 

All the simulations are implemented in PSCAD \cite{website:PSCAD} with 8GB memory and 2.39 GHz CPU.  We use the detailed electromagnetic transient model of three-phase synchronous machines to simulate generators with both governors and exciters.  The uncontrollable load L1-L4 are modelled by the fixed load in PSCAD, while controllable load L1c-L4c are formulated by the self-defined controlled current source. The closed-loop system diagram is shown in Fig.\ref{fig:control2}. We need measure loacal frequency, generation, controllable load and tie-line power flows to compute control demands. Only $\tilde{\phi}_{ij}$ are exchanged between neighbors. All variables are added by their initial steady state values to explicitly show the actual values.
\begin{figure}[htp]
        \centering
        \includegraphics[width=0.45\textwidth]{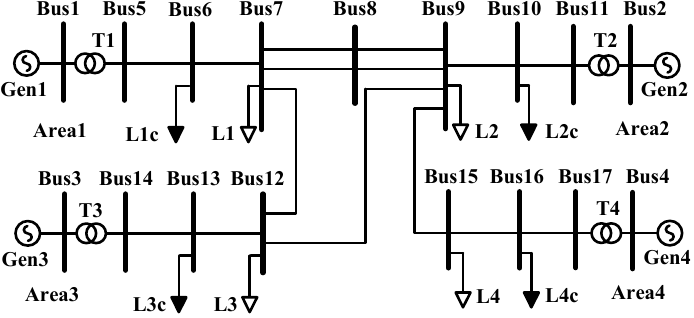}
        \caption{Four-area power system}
        \label{fig:system}
\end{figure}

\begin{figure}[htp]
	\centering
	\includegraphics[width=0.4\textwidth]{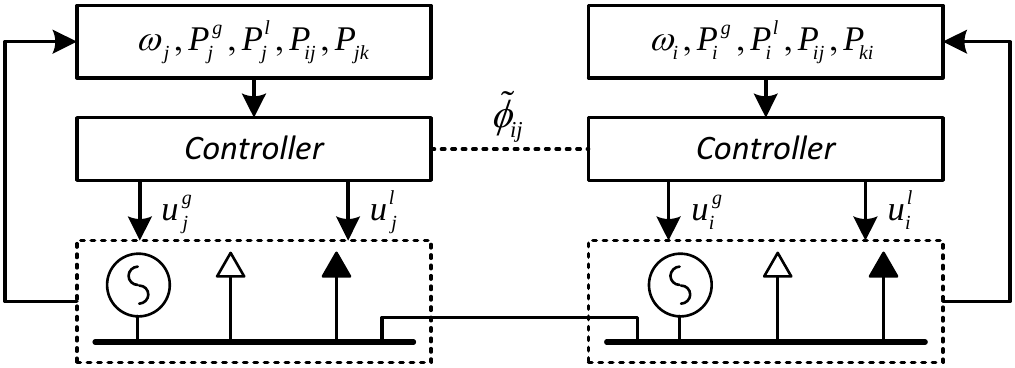}
	\caption{Closed-loop system diagram}
	\label{fig:control2}
\end{figure}

\begin{table}[http]
        \centering
        \caption{\\ \textsc{System parameters}}
        \label{tab:SysPara}
        \begin{tabular}{c c c c c c c}
                \hline 
                Area $j$ & $D_j$ & $R_j$ & $\alpha_j$ & $\beta_j$ &$T^g_j$ &$T^l_j$\\
                \hline
                1 & 0.04  & 0.04  & 2    & 2.5  & 4    & 4 \\
                2 & 0.045  & 0.06  & 2.5  & 4    & 6    & 5 \\
                3 & 0.05  & 0.05  & 1.5  & 2.5  & 5    & 4 \\
                4 & 0.055  & 0.045 & 3    & 3    & 5.5  & 5 \\
                \hline
        \end{tabular}
\end{table}

\subsection {Network Power Balance Case}

In this case, the generations in each area are initiated as (560.9, 548.7, 581.2, 540.6) MW and the controllable loads (70.8, 89.6, 71.3, 79.4) MW. The load changes are identical to those in Table \ref{tab:DistConstraintsnet}, which are also unknown to the controllers. We  use method in Remark 4 to estimate the load changes. Operational constraints on generations, controllable loads and tie lines are shown in Table \ref{tab:DistConstraintsnet}. 

\begin{table}[htb]
	\centering
	\caption{\\ \textsc {Capacity limits in network case}}
	\label{tab:DistConstraintsnet}
	\begin{tabular}{c c c c c}
		\hline 
		& Area 1 & Area 2 & Area 3& Area 4\\
		\hline
		[$\underline{P}^g_j$,$\overline{P}^g_j$] (MW) & [550, 710] & [530, 680] & [550, 700]& [530, 670]\\
		\hline
		[$\underline{P}^l_j$,$\overline{P}^l_j$] (MW) & [20, 80] & [60, 100] & [20, 80] & [35, 80]\\
		\hline
		\hline
		Tie line & (2,1) & (3,1) & (3,2) & (4,2)\\
		\hline
		[$\underline P_{ij}$,$\overline P_{ij}$] (MW) & [-65, 65] & [-65, 65] & [-65, 65] & [-65, 65]\\
		\hline
	\end{tabular}
\end{table}

\subsubsection {Stability and optimality}
The dynamics of local frequencies  and tie-line power flows are illustrated in  Fig.\ref{fig:stabnet}. The frequencies are well restored in all four control areas while the tie line powers are remained within their acceptable ranges. The generations and controllable loads are different from that before disturbance, indicating that the system is stabilized at a new steady state. The resulting equilibrium point is  given in Table IV, which is identical to the optimal solution of  \eqref{eq:opt.2} computed by centralized optimization using CVX. These simulation results confirm that our controller can autonomously guarantee the frequency stability while achieving optimal operating point in the overall system. 

\begin{figure}[htp]
	\centering
	\includegraphics[width=0.5\textwidth]{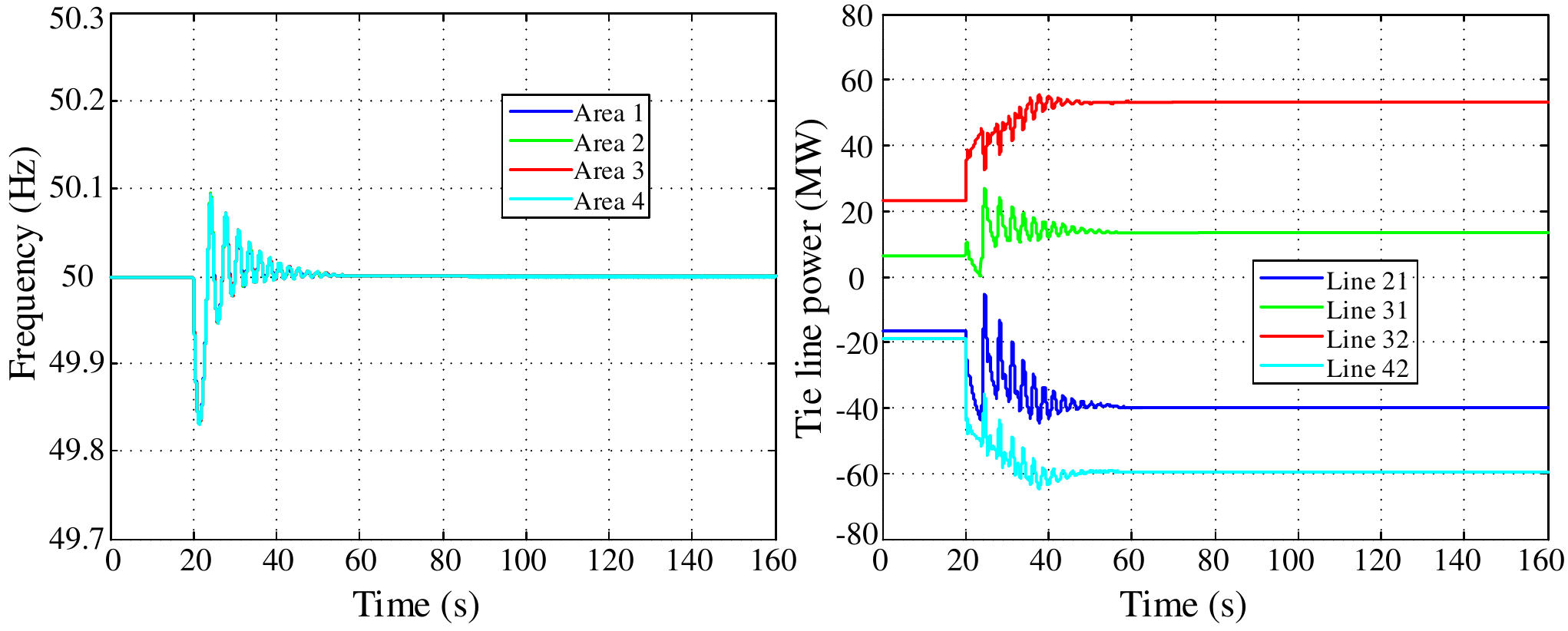}
	\caption{Dynamics of frequency (left) and tie-line flows (right) in network balance case}
	\label{fig:stabnet}
\end{figure}

\begin{table}[htb]
	\centering
	\caption{\\ \textsc {Equilibrium Points}}
	\label{tab:eps}
	\begin{tabular}{c c c c c}
		\hline 
		& Area 1 & Area 2 & Area 3& Area 4\\
		\hline
		$P^{g*}_j$ (MW)  & 620  & 596  & 660  & 580\\
		$P^{l*}_j$ (MW)  & 23.6 & 59.8 & 23.6 & 39.7\\
		\hline
		\hline
		Tie line      & (2,1) & (3,1) & (3,2) & (4,2)\\
		\hline
		$P_{ij}^*$ (MW) & -39.94 & 13.35 & 53.27 & -59.6\\
		\hline
	\end{tabular}
\end{table}

\subsubsection{Dynamic performance}
In this subsection, we analyze the impacts of operational (capacity and  line power) constraints on the dynamic property. Similarly, we compare the  responses of  frequency controllers with and without considering input saturations. The trajectories of mechanical power of turbines and controllable loads are shown in Fig.\ref{fig:dynamicnet.mec} and Fig.\ref{fig:dynamicnet.2}, respectively. In this case, the system frequency is restored, and the same optimal equilibrium point is achieved. 

\begin{figure}[htp]
	\centering
	\includegraphics[width=0.5\textwidth]{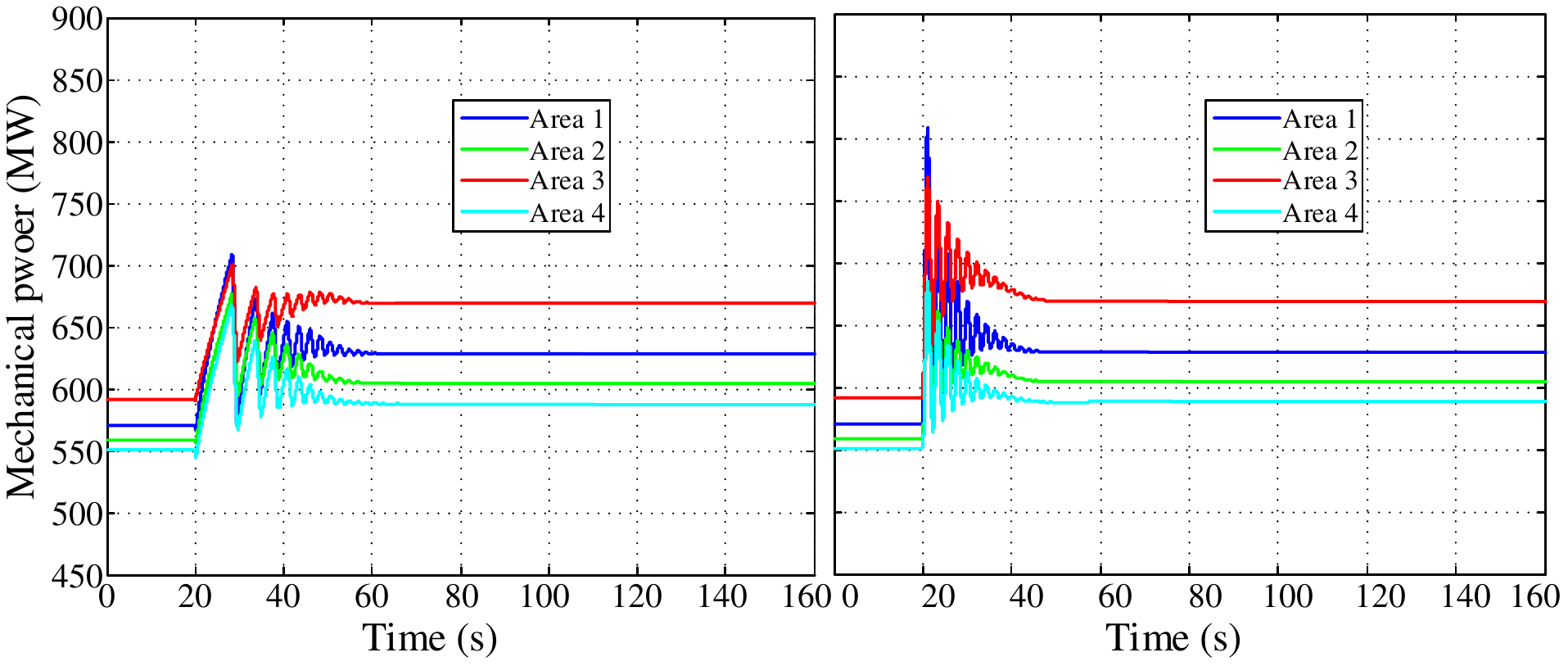}
	\caption{Mechanical outputs with(left)/without(right) capacity constraints}
	\label{fig:dynamicnet.mec}
\end{figure}

%\begin{figure}[htp]
%	\centering
%	\includegraphics[width=0.5\textwidth]{fig7.pdf}
%	\caption{Generators' outputs with(left)/without(right) capacity constraints}
%	\label{fig:dynamicnet.1}
%\end{figure}

\begin{figure}[htp]
	\centering
	\includegraphics[width=0.5\textwidth]{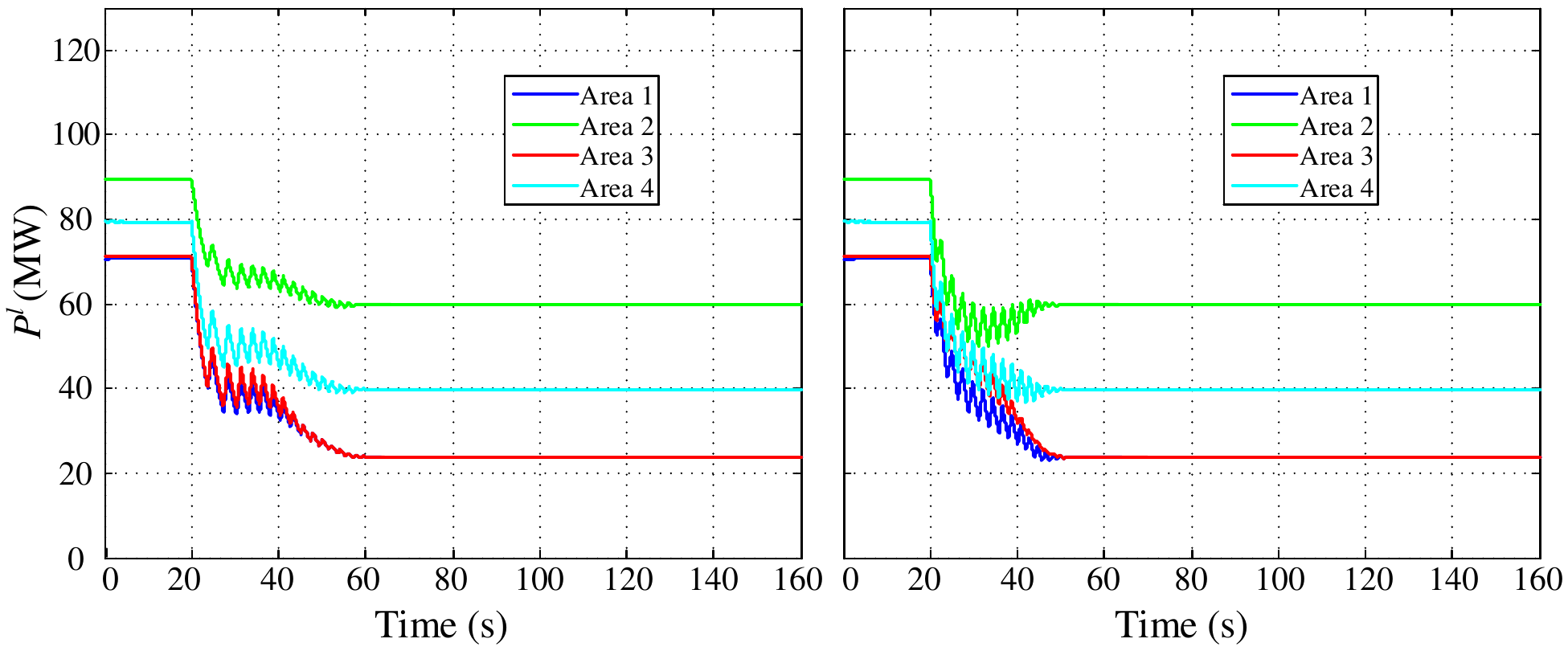}
	\caption{Controllable loads with(left)/without(right) capacity constraints}
	\label{fig:dynamicnet.2}
\end{figure}

\subsubsection {Congestion analysis}
In this scenario, we reduce tie-line power constraints to $\overline{P}_{ij}=-\underline{P}_{ij}=50$MW, which causes congestions in tie-line (2,3) and  (2,4). The steady states under the distributed control are listed in Table \ref{tab:epcong}. Note that $P_j^{g*} - {p_j} + P_j^{l*} - \sum\limits_{k:j \to k} {{P_{jk}^*}}  + \sum\limits_{i:i \to j} {{P_{ij}^*}}  = 0$ hold in each area. 

The dynamics of tie-line powers in two different scenarios shown in Fig.\ref{fig:congestion} indicate that (2,4) reaches the limit in steady state. However, by adopting the proposed fully distributed optimal frequency control, the congestion is eliminated and all the tie line powers remain within the limits.  Thus, congestion control is achieved optimally in a distributed manner.

\begin{table}[htb]
	\centering
	\caption{\\ \textsc {Simulation Results with Congestion}}
	\label{tab:epcong}
	\begin{tabular}{c c c c c}
		\hline 
		& Area 1 & Area 2 & Area 3& Area 4\\
		\hline
		$P^{g*}_j$ (MW) & 618   & 595  & 658  & 585\\
		$P^{l*}_j$ (MW) & 25.1  & 60.7 & 25.1 & 34.9\\
		\hline
		\hline
		Tie line  & (2,1) & (3,1) & (3,2) & (4,2)\\
		\hline
		$P_{ij}^*$ (MW) & -36.4 & 13.1 & 49.5 & -49.9\\
		\hline
	\end{tabular}
\end{table}

\begin{figure}[htp]
	\centering
	\includegraphics[width=0.5\textwidth]{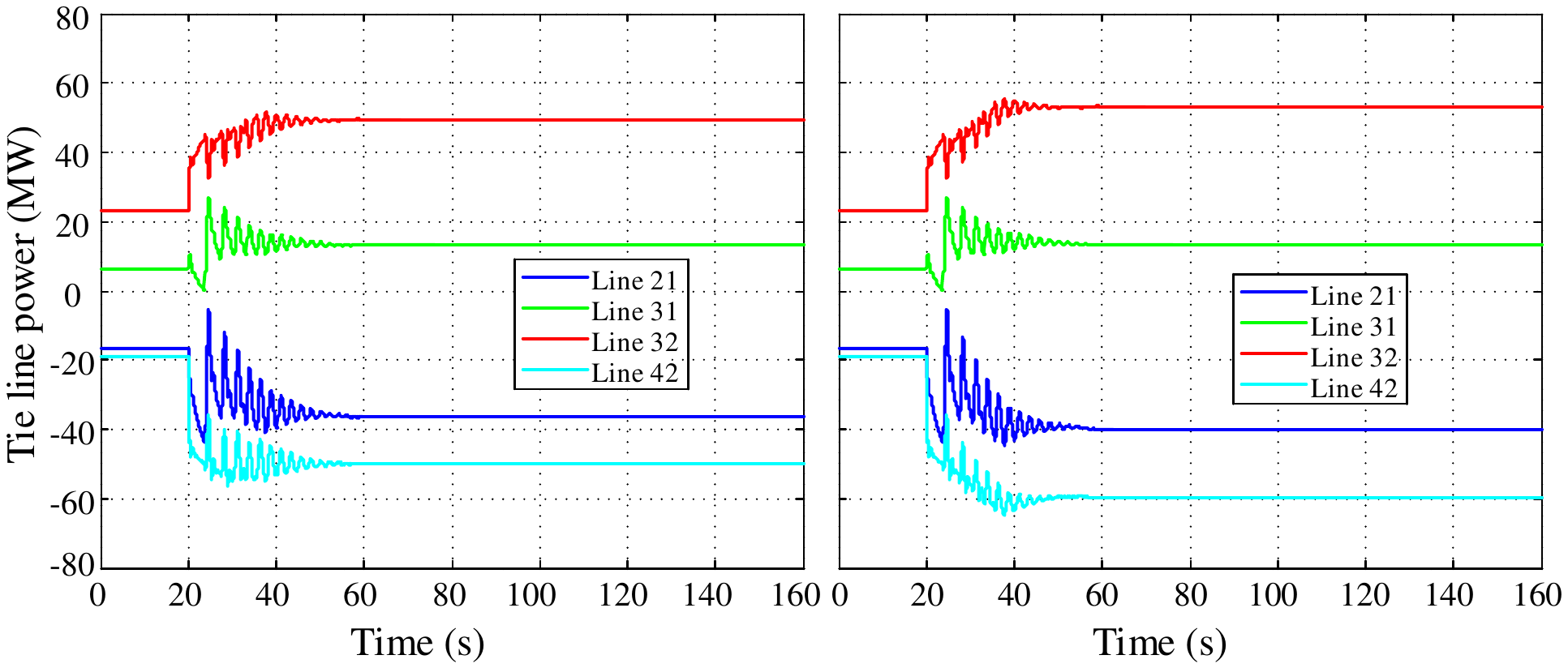}
	\caption{Tie line power with(left)/without(right) capacity constraints}
	\label{fig:congestion}
\end{figure}

\section{Concluding Remarks}
In this paper, we have devised a distributed optimal frequency control in the network balance case, which can autonomously restore the nominal frequencies after unknown load disturbances while minimizing the regulation costs. The capacity constraints on the generations and  controllable loads can also be satisfied even during transient. In addition, congestions can be eliminated automatically, implying tie-line powers can be remained within given ranges. Only neighborhood communication is required in this case. Like the per-node case, here the closed-loop system again carries out a primal-dual algorithm with saturation to solve the associated optimal problem. To cope with the discontinuity introduced due to enforcing different types of constraints, we have constructed a nonpathological Lyapunov function to prove the asymptotically stability of the closed-loop systems. Simulations on a modified Kundur's power system validate the effectiveness of our controller. 

This approach is also applicable to other problem involving frequency regulation, e.g. standalone microgrid or demand side management. We highlight two crucial implications of our work: First, our distributed frequency control is capable of serving as an automatically corrective re-dispatch without the coordination of dispatch center when certain congestion happens; Second, the feasible region of economic decisions can be enlarged  benefiting from the corrective re-dispatch. 
In this sense, our work may provide a systematic way to bridge the gap between  the (secondary) frequency control in a fast timescale  and the economic dispatch in a slow timescale, hence breaking the traditional hierarchy of the power system frequency control and economic dispatch.

\bibliographystyle{IEEEtran}
\bibliography{mybib,PowerRef-201202}

\newpage
\appendices

\makeatletter
\@addtoreset{equation}{section}
\@addtoreset{theorem}{section}
\makeatother

\section{Proofs of Theorem \ref{thm:5} and Theorem \ref{thm:6}}

\renewcommand{\theequation}{A.\arabic{equation}}
\renewcommand{\thetheorem}{A.\arabic{theorem}}

%The rest of this subsection is devoted to the proof of Theorems \ref{thm:5} 
%and Theorem \ref{thm:6}.
We start with a lemma.
\begin{lemma}
	\label{lemma:6}
	Suppose $(x^*, u^*)$ is optimal for \eqref{eq:opt.2}.  Then
	\bee
	\item  $\omega^* = 0$, i.e., the nominal frequency is restored; 
	\item  the network balance condition \eqref{eq:balance.net} is satisfied by $x^*$;
	\item  $\phi^*-\theta^* = (\phi^*_0-\theta^*_0)\,\textbf{1}$;
	\item  $\underline{\theta}_{ij} \le \theta^*_{ij} \le \overline{\theta}_{ij}$, 
		i.e., the line limits \eqref{eq:lineConstraint} are satisfied.
	\eee
\end{lemma}
\begin{proof}
%	The proof of 1) follows a similar argument as that of Lemma \ref{lemma:1},
%	making use of \eqref{eq:opt.2a}\eqref{eq:opt.2b} instead of \eqref{eq:balance.node}.
	Suppose $(x^*, u^*)$ is optimal but $\omega^*\neq 0$.
	Then \eqref{eq:opt.2b} implies
	\bq
	P^{g*} - P^{l*} - p & = &  CBC^T\phi^*
	\label{eq:lemma.6.1}
	\eq
	Consider $\hat x := (\hat\theta, \phi^*, \hat\omega, P^{g*}, P^{l*})$ with
	$\hat\theta:=\phi^*$ and $\hat\omega:=0$.  Then $\hat x$ satisfies
	\eqref{eq:opt.2a}\eqref{eq:opt.2b} due to \eqref{eq:lemma.6.1}.  
	Hence $(\hat x, u^*)$ is feasible for \eqref{eq:opt.2} but has a strictly lower
	cost, contradicting the optimality of $(x^*, u^*)$.  Hence $\omega^*=0$.
	
	Multiplier both sides of \eqref{eq:lemma.6.1} by $\textbf{1}^T$ yields
	the network balance condition \eqref{eq:balance.net}, proving 2).
		
	To prove 3), setting $\omega^*=0$ in \eqref{eq:opt.2a} and combining
	with \eqref{eq:opt.2b} yield
	\bqn
	CBC^T\theta^* &=& P^{g*}-P^{l*}-p \ \ = \ \ CBC^T\phi^*
	\eqn
	Since $CBC^T$ is 
	an $(n+1)\times (n+1)$ matrix with rank $n$, its null space has dimension 1
	and is spanned by $\textbf{1}$ because $C^T\textbf{1} = 0$.
	Hence $CBC^T(\phi^*-\theta^*)=0$ implies that 
	$\phi^*-\theta^* = (\phi^*_0-\theta^*_0)\,\textbf{1}$. 
	%Let 
	%\bqn
	%\tilde \phi :=&\{\phi_{ij},\ \forall %(i,j)\in
	%E\} 
	%\eqn 
	%and
	%\bqn
	%\tilde \theta :=& \{\theta_{ij},\ \forall %(i,j)\in
	%E\}
	%\eqn
	%Then we have $\tilde \phi=C^T\phi$ and $\tilde \theta=C^T\theta$.
	To prove 4), note that $\tilde \phi=C^T\phi$ and $\tilde \theta=C^T\theta$
	and hence 
	\bqn
	\tilde{\phi}^*-\tilde{\theta}^* & = & C^T({\phi^*}-{\theta^*})
			\ \ = \ \ (\phi^*_0-\theta^*_0)C^T\textbf{1}  \ \ = \ \ 0
	\eqn
	i.e. $\tilde{\phi}^*=\tilde{\theta}^*$.   We conclude from \eqref{eq:opt.2c} that
	$\underline{\theta}_{ij} \le \theta^*_{ij} \le \overline{\theta}_{ij}$. 
	This completes 	the proof.
\end{proof}

We have the following result.
\begin{lemma}
	\label{lemma:7}
	Suppose $(x^*, \rho^*)$ is primal-dual optimal.  Then
	\bqn
	u^{g*}_j \ = \
	P^{g*}_j & = & \left[ 
	P^{g*}_j - \gamma^g_j \left( \alpha_j P^{g*}_j + \omega^*_j + z^*_j+\lambda_j^* \right)
	\right]_{\underline P^g_j}^{\overline P^g_j}
	\nonumber \\
	u^{l*}_j \ = \
	P^{l*}_j & = & \left[ 
	P^{l*}_j - \gamma^l_j \left( \beta_j P^{l*}_j - \omega^*_j -z^*_j- \lambda_j^* \right)
	\right]_{\underline P^l_j}^{\overline P^l_j}
	\eqn
	for any $\gamma^g_j>0$ and $\gamma^l_j>0$.
\end{lemma}
\begin{comment}
\begin{proof}
	Since \eqref{eq:opt.2} is (strictly) convex
	with linear constraints, strong duality holds.
	Hence $(x^*,\rho^*) $ is primal-dual optimal
	if and only if it satisfies the KKT condition:\
	$(x^*, u(x^*,\rho^*))$ is primal feasible and
	\bq
	x^*&=&\arg \min_x \{L_2(x;\rho^*)|(x,u(x,\rho^*))
	\ \text{satisfies} \eqref{eq:OpConstraints.1},
	\eqref{eq:opt.2d},\eqref{eq:opt.2e}\}\nonumber\\   \label{eq:lamma.7}
	\eq
	
	The proof is the same as the proof of Lemma
	\ref{lemma:2}, thus is omitted here.
\end{proof}
\end{comment}

Lemma \ref{lemma:7} shows that the saturation of control input does not impact the optimal solution of optimization problem \eqref{eq:opt.2}. 

%\slow{Maybe omit Lemma \ref{lemma:7} since it is not directly
%referenced.  It's used in the detailed proof of the last step of Theorem \ref{thm:5}.}\wang{Agree. Similar argument has been introduced in the Appendix A.}
%\slow{Revisit.}
With Lemma \ref{lemma:bounded.2}, Lemma \ref{lemma:6}
and Lemma \ref{lemma:7}, we now can prove Theorem \ref{thm:5} and Theorem \ref{thm:6}.
\begin{proof}[Proof of Theorem \ref{thm:5}]
	\noindent
	$\Rightarrow$: 
	Suppose $(x^*, \rho^*)$ is primal-dual optimal. 
	Then $x^*$ satisfies the operational constraints \eqref{eq:OpConstraints.1}.
	Moreover the right-hand side of \eqref{eq:model.1} vanishes because:
	\bi
	\item $\dot\theta = 0$ since $\omega^*=0$ from Lemma \ref{lemma:6}.
	\item $\dot\omega = 0$ since constraint \eqref{eq:opt.2a} holds for $x^*$.
	% and $x^*$ achieves network power balance \eqref{eq:balance.net}.
	\item $\dot P^g =\dot P^l=0$ since $\omega^*= 0$ and $x^*$ 
	satisfies \eqref{eq:opt.2d} and \eqref{eq:opt.2e}.
	\item $\dot \lambda = 0$ since \eqref{eq:opt.2b} holds for $x^*$.
	\item $\dot \eta^+ = \dot \eta^- = 0$ since \eqref{eq:opt.2c} holds for $x^*$.
	\item From \eqref{eq:primal.2d} we have 
		\bqn
		\dot {\tilde{\phi}} & = & -\Gamma^{\tilde{\phi}} \ \nabla_{\tilde{\phi}} L_2 (x^*, \rho^*)
		\eqn
		Since $(x^*, \rho^*)$ is a saddle point,
		we must have $\nabla_{\tilde{\phi}} L_2 (x^*, \rho^*)=0$, implying
		$\dot {\tilde{\phi}}=0$.
	\ei 
	Hence $(x^*, \rho^*)$ is an equilibrium of the closed-loop system 
	\eqref{eq:model.1}\eqref{eq:control.2} that satisfies 
	the operational constraints \eqref{eq:OpConstraints.1}.   Moreover 
	$\mu^* = \omega^* = 0$ since
	$
	\frac{\partial L_2}{\partial \omega_j}(x^*, \rho^*)  = 
	D_j ( \omega^*_j - \mu^*_j) = 0
	$
	and $D_j>0$ for all $j\in N$.

	\vspace{0.07in}
	\noindent
	$\Leftarrow$: Suppose now $(x^*, \rho^*)$ is an equilibrium of the closed-loop 
	system \eqref{eq:model.1}\eqref{eq:control.2} that satisfies \eqref{eq:OpConstraints.1}
	and $\mu^*=0$.
	Since \eqref{eq:opt.2} is convex with linear constraints, 
	$(x^*, \rho^*)$ is a primal-dual optimal if and only if 
	$(x^*, u(x^*, \rho^*))$ is primal feasible and satisfies 
	\bq
	x^*& \!\!\!\!\!\!\!\! = \!\!\!\!\!\!\!\! &\arg \min_x \{L_2(x;\rho^*)|(x,u(x,\rho^*))
	\ \text{satisfies} \eqref{eq:OpConstraints.1},
	\eqref{eq:opt.2d},\eqref{eq:opt.2e}\}\nonumber\\  
	\label{eq:lamma.7}
	\eq
	This is because $\nabla_{\rho_1} L_2(x^*, \rho^*)=0$ since
	$\dot\mu = \dot\lambda = 0$, $\eta^{+*}\geq 0$, $\eta^{-*}\geq 0$,
	 and the complementary slackness condition
	$\eta^{+*}_{ij}(\tilde\phi_{ij}^*-\overline{\theta}_{ij})=0$,
	$\eta^{-*}_{ij}(\underline{\theta}_{ij}-\tilde\phi_{ij}^*)=0$ is satisfied
	since $\dot{\eta}^+ = \dot{\eta}^- =0$.
	
	To show that $(x^*, u(x^*, \rho^*))$ is primal feasible, note that since 
	$(x^*, u(x^*, \rho^*))$ is an equilibrium of \eqref{eq:model.1}, it satisfies 
	$\omega^*=0$ and hence \eqref{eq:opt.2d}\eqref{eq:opt.2e}, in
	addition to \eqref{eq:OpConstraints.1}.
	% We have proved in Theorem \ref{thm:1} that $w^*=0$ implies $\theta^*=0$
	% and hence \eqref{eq:model.1b} implies $P^{g*}_j = p_j - P^{l*}_j$ for all $j\in N$,
	% which is \eqref{eq:balance.node}.   
	Moreover $\dot\omega=0$ means $(x^*, u(x^*, \rho^*))$ satisfies
	\eqref{eq:opt.2a}, $\dot \lambda = 0$ implies \eqref{eq:opt.2b}, 
	$\dot \eta^+ = \dot \eta^- = 0$ implies \eqref{eq:opt.2c}.
	
	To show that $(x^*, \rho^*)$ satisfies 
	\eqref{eq:lamma.7}, note that 
	\eqref{eq:opt.2d}\eqref{eq:opt.2e} and \eqref{eq:control.2e}\eqref{eq:control.2f}
	imply that
	\bqn
	P^{g*}_j & = & \left[ 
	P^{g*}_j - \gamma^g_j \left( \alpha_j P^{g*}_j + \omega^*_j +z^{*}_{_{j}} +\lambda_j^* \right)
	\right]_{\underline P^g_j}^{\overline P^g_j}
	\nonumber \\
	P^{l*}_j & = & \left[ 
	P^{l*}_j - \gamma^l_j \left( \beta_j P^{l*}_j - \omega^*_j -z^{*}_{_{j}} - \lambda_j^* \right)
	\right]_{\underline P^l_j}^{\overline P^l_j}
	\eqn
	The rest of the proof follows the same line of argument as that in Theorem 1 in \cite{Wang:DistributedFrequency}.
	This proves that $(x^*,\rho^*)$ is primal-dual optimal and completes the 
	proof of Theorem \ref{thm:5}.
\end{proof}

Next we prove Theorem \ref{thm:6}.
\begin{proof}[Proof of Theorem \ref{thm:6}]
	Let $(x^*, \rho^*) = (\tilde{\theta}^*$$, \tilde \phi^*,$ $
	\omega^*,$ $ P^{g*},$ $ P^{l*},$ $\lambda^*,$$ \eta^{+*},
	\eta^{-*}, \mu^*)$
	be primal-dual optimal. 
%	
%	If $\underline \theta _{ij}<\phi _{ij}<\overline \theta_{ij}$, 
%	Why is this true?  If it does not hold, are $\eta^+, \eta^-$ unique?
%	 $\eta_{ij}^+= \eta_{ij}^-=0 $. 
%	
%	Noticing that $z=0$
%	at optimality, 1) is a direct deduction of $\dot\phi_{ij}=0$.
%	
%	The uniqueness of $x^*, \mu^*$ in 2) is the same as Theorem \ref{thm:1}. 
%	
%	{For $\forall j\in N$, if $P_j^{g*}$ or $P_j^{gl*}$ does not reach the limits, then $\lambda_j^*$ is unique. Why ?}
%	
%	{Otherwise, there must exist $(i,j)\in E$ with $\underline \theta _{ij}<\phi _{ij}^*<\overline \theta_{ij}$. Why?} 
%	
%	Then, $\lambda_j^*=\lambda_i^*$ from 1). If $\sum\nolimits_j(\underline{P}^g_j - \overline{P}^l_j )<\sum\nolimits_j p_j$,  $\sum\nolimits_j(\overline{P}^g_j - \underline{P}^l_j ) > \sum\nolimits_j p_j $, there must exist at least one $j\in N$, where $\underline{P}^g_j < P_j^{g*} < \overline{P}^g_j$ or $\underline{P}^l_j < P_j^{l*} < \overline{P}^l_j$ holds. This implies that $\lambda_j^*$ is unique. Therefore, $\lambda_j^*$ is unique for $\forall j\in N$.		
	For the uniqueness of $x^*$, $(\omega^*, P^{g*}, P^{l*})$ are unique because
	the objective function in \eqref{eq:opt.2} is strictly convex in $(\omega, P^{g}, P^{l})$.
	Hence $\mu^*= \omega^*$ is unique as well.
	Assumption A1 that $\phi^*_0:=0$ and \eqref{eq:lemma.6.1} imply that $\phi^*$ is 
	uniquely determined by the equilibrium $(\omega^*, P^{g*}, P^{l*})$.  
	Since $\theta^* - \phi^* = (\theta^*_0-\phi^*_0)\textbf{1}$, assumption A1 that 
	$\theta^*_0:=0$ then implies that $\theta^*$ is unique.
	This proves the uniqueness of $(x^*, \mu^*)$.
	
	The remaining three parts of the theorem follow from Lemma \ref{lemma:6}.
\end{proof}

\section{Proof Theorem \ref{thm:stability.22}}
\label{appd.thm8}
\renewcommand{\theequation}{B.\arabic{equation}}
\renewcommand{\thetheorem}{B.\arabic{theorem}}

\newcounter{TempEqCnt2}
\setcounter{TempEqCnt2}{\value{equation}}
\setcounter{equation}{1}
\newcounter{mytempeqncnt2}
    \begin{figure*}[!t]
    \normalsize
    \begin{equation}
        \label{eq:Q}
        Q=\\
        	\left[
        	\begin{array}{llllllll}
        		0  & -BC^T   & 0           & 0           & 0        & 0     & 0 & 0 \\
        		CB &   D     & -I_{|N|}    & I_{|N|}     & 0        & 0     & 0 & 0 \\
        		0  &I_{|N|}  & A^g+I_{|N|} & -I_{|N|}    & 0        & 0     & 0 & -CB \\
        		0  &-I_{|N|} & -I_{|N|}    & A^l+I_{|N|} & 0        & 0     & 0 & CB \\
        		0  &0        & -I_{|N|}    & I_{|N|}     & 0        & 0     & 0 & CB \\
        		0  &0        & 0           &   0         & 0        & 0     & 0 & -I_{|E|-|\sigma^+|} \\
        		0  &0        & 0           &   0         & 0        & 0     & 0 & I_{|E|-|\sigma^-|} \\
        		0  &0        & -BC^T        &  BC^T      & -BC^T    &I_{|E|-|\sigma^+|}     &-I_{|E|-|\sigma^-|} & BC^TCB \\
        	\end{array} 
        		\right]
      \end{equation}
        		\hrulefill
        		\vspace*{4pt}
      \end{figure*}  
\setcounter{equation}{\value{TempEqCnt2}}

We start with a lemma.
        
        \begin{lemma}
        	\label{lemma:decreasing.2}
        	Suppose A1, A4 and A5 hold. Then 
        	\begin{enumerate}
        		\item $\dot V_2(w(t))\leq 0 $ in a fixed $\sigma^+$ and $\sigma^-$.
        		\item The trajectory $w(t)$ is bounded, i.e., there exists $\overline w$ such
        		that $\|w(t)\|\leq \overline w$ for all $t\geq 0$.
        	\end{enumerate}  	
%        	$V_n(v)$ is always decreasing along system \eqref{eq:model.4}.
        \end{lemma}
        
\begin{proof}[Proof of Lemma \ref{lemma:decreasing.2} ] 
	Given fixed $\sigma^+$, $\sigma^-$, for all $(i,j)\notin \sigma^+$, $(i,j)\notin \sigma^-$, we have
	\begin{subequations}
		\begin{align}
		\dot V_2(w)&\le k(H(w)-w^*)^T\cdot \Gamma_2^{-1} (H(w)-(w-F(w)))	\label{eq:ProofThm.2a}	\\
		& -(H(w)-w)^T\Gamma_2\cdot Q\cdot\Gamma_2(H(w)-w)\label{eq:ProofThm.2b}\\
		& - \left( H(w) - (w - F(w)) \right)^T(\Gamma_2-k\Gamma_2^{-1}) (w-H(w))  \label{eq:ProofThm.2c}\\
		& -k(w-w^*)^T\cdot \Gamma_2^{-1} F(w)	\label{eq:ProofThm.2d}
		\end{align}
		where $Q$  is a semi-definite positive matrix and $\Gamma_2 Q=\nabla_wF(w)$, which given in \eqref{eq:Q}. Here, the subscript
		of $I$ means its dimension, and $|A|$ means the cardinality of set $A$. 	           
	\end{subequations}
	
%	\fliu{In the previous definition, we consider all $(i,j)$ in $V_2$. But here we exclude those $(i,j)\in\sigma^+$ and $(i,j)\in\sigma^-$. Should we add some description to make them  consistent?   }
	
	Given fixed $\sigma^+$ and $\sigma^-$, $F(w)$ is continuous differentiable. In this case, (\ref{eq:ProofThm.2a}) and (\ref{eq:ProofThm.2c}) are nonpositive due to as discussed in the proof of Theorem 4 in \cite{Wang:DistributedFrequency}. (\ref{eq:ProofThm.2b}) is also nonpositive as $Q$ is semi-definite positive (see Eq. \eqref{eq:Q}).     
	
	Next, we  prove that (\ref{eq:ProofThm.2d}) is nonpositive. 
	Similar to the per-node case, substituting $\mu(t)\equiv \omega(t)$ into the Lagrangian $L_2(x, \rho)$ in \eqref{eq:defL.21} we obtain a function $\hat{L}_2(w)$ defined as follows.  
%	\slow{What is the difference between $\hat L_2$ and $L_2$ below?}\wang{There is no $\mu$ in $\hat{L}_2$, similar to $\hat{L}_1$.}
	\begin{align}
	%	    \hat{L}(w)&\quad =\hat L_2(\tilde {\theta}, \omega, P^g, P^l, \lambda, \eta^+, \eta^-, \tilde{\phi})  \nonumber\\   
	\hat{L}_2(w)& \quad := 
	L_2 \left(\tilde {\theta}, \omega, P^g, P^l, \lambda, \mu, \eta^+, \eta^-, \tilde{\phi} \right)_{\mu=\omega} 
	\nonumber
	\\ &\quad =  
	\frac{1}{2} \left( (P^g)^T A^g P^g + (P^l)^T A^l P^l - \omega^T D \omega+z^Tz \right )\nonumber
	\\
	& \quad + \ \lambda^T \!\! \left( P^g - P^l - p -CB\tilde{\phi}\right) + \ \omega^T \!\! \left( P^g - P^l - p - C B\tilde \theta \right)\nonumber\\
	&\quad +(\eta^-)^T\left(\underline{\theta}-\tilde\phi\right) + (\eta^+)^T \left(\tilde \phi-\overline{\theta}\right)\nonumber
	\end{align}
	In addition, denote $w_1=(\tilde\theta, P^g, P^l, \tilde{\phi})$, $w_2=(\lambda, \omega, \eta^+, \eta^-)$. Then $\hat L_2$ is convex in $w_1$ and concave in $w_2$.  
	
	In $F(w)$,  $[\tilde\phi-\overline \theta]^+_{\eta^+}$ and $[\underline \theta-\tilde\phi]^+_{\eta^-} $ have unknown dimensions (up to $\sigma^+$ and $\sigma^-$). Fortunately, we have 
	\begin{align}
	(\eta^+-\eta^{+*})^T[\tilde\phi-\overline \theta]^+_{\eta^+}&\le(\eta^+-\eta^{+*})^T(\tilde\phi-\overline \theta)\nonumber\\
	&=(\eta^+-\eta^{+*})^T\nabla_{\eta^+} \hat L_2\nonumber
	\end{align}
	where the inequality holds since $\eta_{ij}^{+}=0\le \eta_{ij}^{+*}$ and $\tilde\phi_{ij}-\overline \theta_{ij} < 0$ for $(i,j)\in \sigma^+$, i.e., $(\eta_{ij}^+-\eta_{ij}^{+*}) \cdot (\tilde{\phi}_{ij}-\overline{\theta}_{ij})\ge0$. Similarly, 
	\begin{align}
	(\eta^--\eta^{-*})^T[\underline \theta-\tilde\phi]^+_{\eta^-}&\le(\eta^--\eta^{-*})^T(\underline \theta-\tilde\phi) \nonumber\\
	& = (\eta^--\eta^{-*})^T\nabla_{\eta^-} \hat L_2.\nonumber
	\end{align}
	
	Consequently, it can be verified that 
	\bqn
	(w-w^*)^T\Gamma_2^{-1} F(w) & \!\!\!\! \le \!\!\!\! & (w-w^*)^T \nabla_{w}\hat{L}_2(w) \nonumber\\
	& \!\!\!\!= \!\!\!\! &\ 
	(w-w^*)^T\begin{bmatrix} \ \ \nabla_{w_1}\hat L_2 \\ - \nabla_{w_2} \hat L_2 \end{bmatrix}\!
	(w_1, w_2)
	\eqn
	where, $\nabla_{w_1} \hat{L}_2 := \begin{bmatrix}
	\nabla_{\tilde\theta} \hat L_2 \\
	\nabla_{P^g} \hat L_2 \\
	\nabla_{P^l} \hat L_2 \\
	\nabla_{\tilde{\phi}} \hat L_2 
	\end{bmatrix}$   and $\nabla_{w_2} \hat{L}_2 := \begin{bmatrix}
	\nabla_{\lambda} \hat L_2 \\
	\nabla_{\omega} \hat L_2\\
	\nabla_{\eta^+} \hat L_2\\
	\nabla_{\eta^-} \hat L_2
	\end{bmatrix}$.  
	
	\setcounter{equation}{2}
	Then we have
	\begin{align}
	&-k(w-w^*)^T\Gamma_2^{-1} F(w)\le \ -k(w-w^*)^T\cdot \Gamma_2^{-1} F(w) \nonumber \\
	= &\  -k(w_1-w_1^*)^T \nabla_{w_1}\hat L_2(w_1,w_2) + k(w_2-w_2^*)^T \nabla_{w_2}\hat L_2(w_1,w_2) \nonumber \\
	\le & \ k \left(\hat L_2(w_1^*,w_2) - \hat L_2(w_1,w_2) +\hat L_2(w_1,w_2) - \hat L_2(w_1,w^*_2)\right) \nonumber \\
	= &  \ k\left( \underbrace{\hat L_2(w_1^*,w_2) - \hat L_2(w^*_1,w^*_2)}_{\le 0} + \underbrace{\hat L_2(w^*_1,w^*_2) - \hat L_2(w_1,w^*_2)}_{\le 0}\right)\nonumber \\
	\le & \ 0
	\label{saddle point}	
	\end{align}
	where the first inequality holds because $\hat L_2$ is convex in $w_1$ and concave in $w_2$ and the second inequality follows
	because $(w^*_1, w^*_2)$ is a saddle point.    
	Therefore (\ref{eq:ProofThm.2d}) is nonpositive, proving the first assertion. 
	
	To prove the second assertion, we further investigate the situation that $\sigma^+$ or  $\sigma^-$ changes. 
	We only consider the set $\sigma^+$ since it is the same to  $\sigma^-$. We have the following observations:
	\begin{itemize}
		\item The set $\sigma^+$ is reduced, which only happens when $\tilde{\phi}_{ij}-\overline{\theta}_{ij}$ goes through zero, from negative to positive. Hence an extra term will be added to $V_2$. As this term is initially zero, there is no discontinuity of $V_2$ in this case.
		\item The set $\sigma^+$ is enlarged when $\eta_{ij}^+$ goes to zero from positive while $\tilde{\phi}_{ij}<\overline{\theta}_{ij}$. Here $V_2$ will lose a positive term $(\gamma_{ij}^\eta)^2(\tilde{\phi}_{ij}-\overline \theta_{ij})^2/2$, causing discontinuity.
	\end{itemize}

	In the context, we conclude that $V_2$ is always nonincreasing along the trajectory even when $\sigma^+$ or $\sigma^-$ changes and discontinuity occurs. 

	To prove that the trajectory $w(t)$ is bounded
	note that \cite[Theorem 3.1]{Fukushima:Equivalent} proves that 
	$\hat V_2(w) := - \left( H(w)-w \right)^T F(w)  \, - \, \frac{1}{2} ||H(w)-w||^2_2$ satisfies
	$\hat V_2(w) \geq 0$ over $S$. 	
	Therefore, we have 
	\bqn
	\frac{1}{2}k(w(t)-w^*)^T\Gamma_2^{-2}(w(t)-w^*) &\!\!\! \le \!\!\! & V_2(w(t)) \ \le \ V_2(w(0))
	\eqn
	indicating the trajectory $w(t)$ is bounded.  
\end{proof}

\begin{lemma}
        		\label{lemma:decreasing.3}
        		Suppose A1, A4 and A5 hold. Then 
        		\begin{enumerate}
        			\item The trajectory $w(t)$  converges to the largest weakly invariant subset $W_2^*$ contained in $W_2:=\{w\in S | \dot V_2(w)=0 \}$.
        			\item Every point $w^*\in W_2^*$ is an equilibrium point of \eqref{eq:model.4}.
        		\end{enumerate}  	
        		%        	$V_n(v)$ is always decreasing along system \eqref{eq:model.4}.
        	\end{lemma}
        
        \begin{proof} [Proof of Lemma \ref{lemma:decreasing.3} ] 
        	Given an initial point $w(0)$ there is a compact set $\Omega_0 := \Omega(w(0)) \subset S$ such that $w(t)\in\Omega_0$ 
        	for $t\geq 0$ and $\dot V_2(w) \leq 0$ in $\Omega_0$. 
        	
        	Invoking the proof of Lemma \ref{lemma:decreasing.2}, $V_2$ is radially unbounded and positively definite except at equilibrium. As $V_2$ and $\dot{V}_2$ are nonpathological,   we conclude that any trajectory $w(t)$ starting from $\Omega_0$ converges to the largest weakly invariant subset $W_2^*$ contained in $W_2=\{\ w\in \Omega_0\ |\ \dot V_2(w) =0\ \}$ \cite[Proposition 3]{Bacciotti:Nonpathological}, proving  the first assertion.
        	
        For the second assertion, We fix $w(0) \in W_2^*$ and then prove that $w(0)$ must be an equilibrium point. 
        		
       	From (\ref{eq:ProofThm.2b}), direct computing yields 
        	\begin{align}
        	\dot V_{2}(w(t)) &\le  -\dot{\omega}^TD \dot {\omega}-(\dot {P}^g)^T A^g \dot {P}^g \nonumber\\
        	&-(\dot {P}^l)^T A^l \dot {P}^l -(CB\dot{\tilde \phi}-\dot {P}^g+\dot {P}^l)^T(CB\dot{\tilde \phi}-\dot {P}^g+\dot {P}^l)\nonumber\\
        	&\le 0
        	\label{dot V2}
        	\end{align}
       Since $A^g$, $A^l$ and $D$ are positively definite diagonal matrices, $\dot{V}_2(w)=0$ holds  only when $ \dot P^g = \dot P^l = \dot \omega=0$. Therefore, for any $w(0)  \in W_2^*$, the trajectory $w(t)$ satisfies 
       \bq
       \label{eq:dotV=0}
       \dot P^g(t) \ = \ \dot P^l(t) \ = \ \dot \omega(t) \ = \ 0, \qquad t\ge 0
       \eq
Hence $P^g(t)$, $P^l(t)$ and $\omega(t)$ are all constants due to the boundedness property guaranteed by  Lemma \ref{lemma:decreasing.2}.

   On the other hand, for $\dot V_2(w)=0$, both terms in  (\ref{saddle point}) have to be zero, implying that 
        	$$\hat L_2(w_1(t),w_2^*) = \hat L_2(w^*_1,w^*_2)$$
        	must hold in $W_2$. Differentiating with respect to $t$ gives
        	\begin{align}
	        	\left(\frac{\partial}{\partial w_1}\hat L_2(w_1(t),w_2^*)\right)^T\cdot\dot w_1(t)=0        	=-\dot{\tilde{\phi}}^T(\Gamma^{\tilde \phi}) ^{-1}\dot{\tilde{\phi}}
        	\end{align}
        	The second equality holds due to Eq. \eqref{eq:dotV=0} and \eqref{eq:closedloop.2h}. Then we can conclude $\dot{\tilde{\phi}}=0$ immediately, implying $\tilde{\phi}$ is also constant in $W^*$ due to its boundedness.        	
        	
%        	Let $W_2 :=\{w\in \Omega_0 | \ \dot P^g= \dot P^l= \dot \omega=\dot{\tilde \phi} =0\}$.
%        	Then \eqref{dot V2} and \eqref{dot V22} imply that $w\in W_2$ if and only if $\dot V_2(w) = 0$.
%        	
%        	
%        	For the second assertion, fix any $w(0)\in W_2^*$. 
	Invoking the close-loop dynamics  \eqref{eq:closedloop.2},  $\dot{\tilde{\theta}}(t)$, $\dot{\eta}^+(t)$, $\dot{\eta}^-(t)$ and $\dot{\lambda}(t)$ must be constants in $W_2^*$ as $P^{g}(t), P^l(t), \omega(t)$ and ${\tilde{\phi}}(t)$ are all constants. Then we conclude that $\dot{\tilde{\theta}}(t)= \dot{\eta}^+(t)=\dot{\eta}^-(t)=\dot{\lambda}(t)=0$ holds for all $t\ge 0$ due to the boundedness property of $w(t)$ (Lemma \ref{lemma:decreasing.2}).
       	This implies that any $w(0)\in W_2^*$ must be an equilibrium point,  completing the proof.
        \end{proof}
   
\begin{proof}[Proof of Theorem \ref{thm:stability.22} ] 
		 Fix any initial state $w(0)$ and consider the trajectory $(w(t), t\geq 0)$ of the
		 closed-loop system close-loop dynamics  \eqref{eq:closedloop.2}.		 
		 As mentioned in the proof of Lemma \ref{lemma:decreasing.3},
		 $w(t)$ stays entirely in a compact set $\Omega_0$.   Hence there exists an infinite 
		 sequence of time instants ${t_k}$ such that $w(t_k)\to \hat {w}^*$ as $t_k\to\infty$, 
		 for some $\hat w^* \in W_2^*$. Lemma \ref{lemma:decreasing.3} guarantees that
		 $\hat w^*$ is an equilibrium point of the closed-loop system \eqref{eq:closedloop.2}, and hence $\hat w^*=H(\hat{w}^*)$.
		  Thus, using this specific equlibrium point $\hat {w}^*$ in the definition of $V_2$, we have 
%		 \begin{align}
%		 	V_2^* = \lim_{t\to \infty} V_2(t) = \lim_{t_k\to \infty} V_2(t_k) = V_2(\lim_{t_k\to \infty} w(t_k)) = V_2(\hat w^*) = 0
%		 	\nonumber
%		 \end{align}
		 \begin{align}
			 V_2^* =\lim\limits_{t\to \infty} V_2(w(t)) &= \lim\limits_{t_k\to \infty} V_2(w(t_k)) \nonumber\\
				 &\qquad=\lim\limits_{w(t_k) \to \hat w^*} V_2\big( w(t_k)\big) = V_2(\hat w^*) = 0 \nonumber
		 \end{align}
		 Here, the first equality uses the  fact that
		 $V_2(t)$ is nonincreasing in $t$; the second equality uses the fact that
		 $t_k$ is the infinite sequence of $t$; the third equality uses the fact that $w(t)$ is
		 absolutely continuous in $t$; the fourth equality is due to the upper semi-continuity of $V_2(w)$, and the last equality holds as $\hat {w}^*$ is an equilibrium point of $V_2$. 
		 
		 The quadratic term $(w-\hat w^*)^T\Gamma_2^{-2}(w-\hat w^*)$
		 in $V_2$ then implies that $w(t)\to \hat {w}^*$ as $t\to \infty$, which completes the proof. 
\end{proof}

\end{document}